\documentclass[journal,draftcls,onecolumn,12pt,twoside]{IEEEtranTCOM}
\usepackage{tikz}
\usepackage{acronym}
\usetikzlibrary{shapes,snakes}
\usepackage{empheq}
\usepackage{graphicx}
\usepackage{epstopdf}
\usepackage{multirow}
\usepackage{tabu}
\usepackage{array}
\usepackage{url}
\usepackage{amsmath}
\usepackage{cleveref}
\usepackage{amssymb}
\usepackage{amsthm}
\usepackage{epstopdf}
\usepackage{amsfonts}
\usepackage{bm}
\usepackage{algorithmic}
\usepackage[ruled,vlined,linesnumbered]{algorithm2e}
\usepackage{graphicx}
\usetikzlibrary{arrows}
\usepackage{cite}
\usepackage{caption}
\usepackage{subcaption}
\usepackage{amssymb}
\usepackage{tabulary}
\usepackage{booktabs}
\usepackage{soul}

\newcommand{\TC}[1]{\textcolor{cyan}{#1}}

\newcommand{\revtwo}[1]{\textcolor{black}{#1}}

\RequirePackage{xcolor}
\def\current@color{ Black}

\newcommand{\them}[1]{\textcolor{black}{#1}}

\usepackage[justification=centering]{caption} 
\DeclareCaptionLabelFormat{algnonumber}{algorithm}
\tikzstyle{int}=[draw, fill=white!20, minimum size=2em]
\tikzstyle{init} = [pin edge={to-,thin,black}]

\newtheorem{thm}{Theorem}

\newtheorem{defi}{Definition}
\newtheorem{lemma}{Lemma}

\newcommand\blfootnote[1]{%
	\begingroup
	\renewcommand\thefootnote{}\footnote{#1}%
	\addtocounter{footnote}{-1}%
	\endgroup
}

\DeclareRobustCommand*{\IEEEauthorrefmark}[1]{%
	\raisebox{0pt}[0pt][0pt]{\textsuperscript{\footnotesize\ensuremath{#1}}}}

\acrodef{AoI}{Age of Information}
\acrodef{IoT}{Internet of Things}
\acrodef{DPP}{Drift-Plus-Penalty}
\acrodef{DTMC}{Discrete-Time Markov Chain}
\acrodef{w.p.}{with probability}
\acrodef{EVD}{EigenValue-Decomposition}
\acrodef{MDP}{Markov Decision Process}
\acrodef{CSI}{Channel State Information}
\acrodef{CMDP}{Constrained Markov Decision Process}
\acrodef{i.i.d}{independent and indentically distributed}
\acrodef{psDPP}{per-slot Drift-Plus-Penalty}
\acrodef{pfDPP}{per-frame Drift-Plus-Penalty}

\IEEEoverridecommandlockouts                          

\title{
	
	Scheduling Policies for AoI Minimization with Timely Throughput Constraints 
}

\author{
	\IEEEauthorblockN{Emmanouil Fountoulakis, 
		Themistoklis Charalambous, \textit{Senior Member, IEEE},
		Anthony Ephremides\textit{, Life Fellow, IEEE},
		Nikolaos Pappas, \textit{Senior Member, IEEE}}
}

\begin{document}
	\immediate\write18{echo $PATH > tmp1}
	\immediate\write18{/Library/TeX/texbin/epstopdf > tmp2} 
	
	\maketitle
	\thispagestyle{plain} 
	\pagestyle{plain}
	
	\begin{abstract}
		\blfootnote{
			E. Fountoulakis was with the Department of Science and Technology, Linköping University, Norrköping, Sweden. He is now with Ericsson AB, Sweden (email: emmanouil.fountoulakis@ericsson.com). 
			
			N. Pappas is with the Department of Computer and Information Science, Linköping University, Linköping, Sweden (nikolaos.pappas@liu.se).
			
			T. Charalambous is with the Department of Electrical and
			Computer Engineering, School of Engineering, University of Cyprus. Email: \{themistoklis.charalambous@aalto.fi\}.
			
			A. Ephremides is with the Electrical and Computer Engineering Department, University of Maryland, College Park, USA. Email: \{etony@umd.edu\}.
			
			The work of N. Pappas has been supported in part by the Swedish Research
			Council (VR), ELLIIT, Zenith, and the European Union (ETHER, 101096526).
			
			The work of T. Charalambous was supported in part by the European
			Research Council (MINERVA, 101044629)
		}
		
		In 5G and beyond communication systems, the notion of latency gets great momentum in wireless connectivity as a metric for serving real-time communications requirements. However, in many applications, research has pointed out that latency could be inefficient to handle applications with data freshness requirements. Recently, Age of Information (AoI) metric, which can capture the freshness of the data, has attracted a lot of attention. 
		In this work, we consider mixed traffic with time-sensitive users; a deadline-constrained user, and an AoI-oriented user. To develop an efficient scheduling policy, we cast a novel optimization problem formulation for minimizing the average AoI while satisfying the timely throughput constraints. The formulated problem is cast as a Constrained Markov Decision Process (CMDP). We relax the constrained problem to an unconstrained Markov Decision Process (MDP) problem by utilizing the Lyapunov optimization theory and it can be proved that it is solved per frame by applying backward dynamic programming algorithms with optimality guarantees. In addition, we provide a low-complexity algorithm guaranteeing that the timely-throughput constraint is satisfied. Simulation results show that the timely throughput constraints are satisfied while minimizing the average AoI. Simulation results show the convergence of the algorithms for different values of the weighted factor and the trade-off between the AoI and the timely throughput.
	\end{abstract}
	
	%
	
	\IEEEpeerreviewmaketitle 
	\section{Introduction}
	
	With the advent of 5G communication networks, the metric of latency plays a vital role in wireless connectivity for addressing the requirements of real-time communications, such as autonomous vehicles, wireless industrial automation, environmental, and health monitoring, to name a few \cite{abd2019role, shreedhar2019age}. In real-time communications, information is required to arrive at the destination within a certain period (deadline-constrained) due to stringent requirements in terms of latency, while in other cases, it is required to keep the information at the destination as fresh as possible. The notion of packets with deadlines is connected with the \textit{timely-throughput}, that is \textit{the average number of successful packet deliveries before their deadline expiration}  \cite{TheoryQoS2009}. Information ``freshness" is captured by a new metric called \ac{AoI} \cite{kosta2017age, sun2019age}. It was first introduced in \cite{kaulyates2012real}, and it is defined as \textit{the time elapsed since the generation of the status update that was most recently received by a destination}. Furthermore, time-sensitive applications with different requirements co-exist in the same network and share the same resources. Therefore, it is important to allocate the resources efficiently in order to satisfy the requirements of the heterogeneous traffic. 
	
	In order to \them{enhance} our understanding \them{of} such systems, we consider a system with an \ac{AoI}-oriented user and a user with timely throughput requirements, which can be considered two of the main categories in time-sensitive networking. Our goal is to minimize the average \ac{AoI} while satisfying the timely throughput requirements. We consider both time-correlated channel model and \ac{i.i.d} channel model. We utilize tools from Lyapunov optimization theory in order to transform the initial \ac{CMDP} problem into a \ac{MDP} problem, and we provide a dynamic programming algorithm that solves the problem optimally. Furthermore, we provide a low-complexity algorithm that guarantees that the timely throughput constraints are satisfied.
	
	Wireless systems with packets with deadlines have been considered almost two decades ago \cite{shakkottai2002scheduling}. An extensive survey that provides an overview of the mathematical tools that are used in the area of resource control for delay-sensitive networks can be found in \cite{cui2012survey}. Recently, there has been a renewed interest in studying the performance of systems with deadline-constrained traffic \cite{you2018resource, ManosWiOpt2017, ElAzzouni2020, tsanikidis2021power,destounis2018scheduling,tsanikidis2022randomized}, especially due to the ongoing automation of traditional manufacturing and industrial practices under the fourth industrial revolution. Packets with deadlines are connected with the notion of \textit{timely throughput}. Timely throughput is first introduced in \cite{TheoryQoS2009}, and it is defined as the average number of successfully delivered packets before their deadline expiration. In \cite{TheoryQoS2009}, the authors propose an algorithm that can satisfy any feasible timely throughput constraint. An extensive study of the timely throughput of heterogeneous wireless networks is provided in \cite{lashgari2013timely}. 
	
	Real-time scheduling optimization for deadline-constrained traffic in wireless systems has been extensively studied in the literature. In \cite{ManosGC2018}, the authors provide a dynamic algorithm for minimizing the packet drop rate while satisfying the average power constraints. In \cite{TepedeTVT2018}, the authors consider a joint scheduling and power allocation problem for a network with real-time traffic, i.e., packets with deadlines, and non-real time traffic. Furthermore, a mixed type of traffic is considered in \cite{fountoulakisITU}. The authors consider a joint scheduling and power allocation problem for a network with deadline-constrained users and users with minimum throughput requirements. Furthermore, it has been often shown that \them{it} is natural to formulate this kind of problem as \them{an} \ac{MDP} \cite{master2016power, neely2013dynamic, ModianoTWC2006}. For example, in \cite{neely2013dynamic}, the authors formulate the problem of minimizing the packet drop rate while providing queueing stability as an \ac{MDP} problem. Several approximations and also a methodology are provided that can be applied to general problems of this kind. In addition, there are works that consider packets with deadlines in multi-hop networks \cite{TNSM2020, mao2014optimal,tsanikidis2022online,gu2021asymptotically, deng2019online}.

	
	Recently\them{,} the optimization and control of average or peak \ac{AoI} has  attracted a lot of attention for a plethora of scenarios \cite{talak2020age, kadota2021minimizing, kadota2019scheduling,zhou2019joint,abd2020aoi, BedewyOptimalSampling,fountoulakis2021joint,yao2020age, ceran2019reinforcement, GstamIoTJ}. There are two cases for generation of the status updates: i) status updates arrive randomly at the users \cite{talak2020age, kadota2021minimizing}, ii) \textit{generate-at-will} mechanism that allows the user to sample fresh data at will \cite{ kadota2019scheduling,zhou2019joint,abd2020aoi}. In \cite{talak2020age, kadota2021minimizing}, the authors consider the problem of \ac{AoI} minimization in single-hop networks with stochastic arrivals and they provide several scheduling policies. In \cite{kadota2019scheduling}, the problem of \ac{AoI} minimization with throughput constraints in a multi-user network is considered. In \cite{BedewyOptimalSampling}, the problem of joint sampling and scheduling is considered for multi-source systems. Furthermore, there is a line of works that considers transmission and sampling costs \cite{fountoulakis2021joint, zhou2019joint}. In these works, if the sampled information failed to be transmitted due to channel errors, it may be retransmitted in next slots to reduce the sampling cost. In \cite{yao2020age}, the \ac{AoI} minimization problem with energy constraints is considered. It is shown that the optimal policy is a mixture of two stationary deterministic policies. 
	In \cite{ceran2019reinforcement,abd2020aoi, GstamIoTJ}, the optimization of \ac{AoI} in \ac{IoT} and energy harvesting systems has been studied. In these works, the problems are formulated as \acp{MDP} and they are solved by using tools from dynamic programming and reinforcement learning.
	
	Although there are many works that consider the \ac{AoI} optimization or analysis, there are few works that consider \ac{AoI} optimization in a system with heterogeneous traffic \cite{ZhengOpenJournal, SPAWC2019, fountoulakis2022information, sun2021age}. 
	The work that is closer to our work is \cite{sun2021age}. The authors in \cite{sun2021age} consider a wireless network including \ac{AoI}-oriented users and deadline-constrained users. The goal is to minimize the average \ac{AoI} while satisfying the timely throughput constraints. The authors \them{in \cite{sun2021age} also} consider that the time is divided into frames and the frames into slots. \them{However, they additionally assume}  that the \ac{AoI}-oriented user can be scheduled in any time slot within the frame and the value of the \ac{AoI} remains $1$, if the transmission succeeds, during the whole frame. Furthermore, it is \them{assumed} that the channel remains fixed during a frame. \them{On the contrary}, we consider that \ac{AoI} is  $1$ only when the \ac{AoI}-oriented user \them{transmits a packet successfully}. \them{Furthermore},  the channel of \them{a} user can change from slot to slot \them{unlike} from frame to frame. These assumptions make the problem \them{considered in our paper} fundamentally different and more realistic.  
	
	\subsection{Contributions}
	In this work, we consider two users that send their information over an error-prone channel to a common receiver. The first user is \ac{AoI}-oriented and the second user has timely throughput requirements. We consider two-channel model cases: i) i.i.d. channels over time slots, ii) time-correlated channels. Our goal is to minimize the average \ac{AoI} while satisfying the timely throughput requirements. The problem is formulated as a \ac{CMDP} problem which is known to be a difficult problem to solve and standard approaches, such as the method of Lagrange multipliers, cannot be directly applied. To solve this problem, we first apply tools from Lyapunov optimization theory to transform the \ac{CMDP} into an \ac{MDP}. It is shown that the infinite \ac{CMDP} can be reduced to an unconstrained weighted stochastic shortest path problem, i.e., a finite-horizon \ac{MDP}, that is easier to be solved. Note that our approach can be applied in more general stochastic optimization problems of that type. The relaxed problem is then solved by invoking backward dynamic programming. Simulation results show that the optimal policy schedules the \ac{AoI} user multiple times within a frame with some probability. That means that the optimal decision is not to schedule the \ac{AoI} only at specific slots of the frame with high probability, e.g., at the beginning or the end of the frame. Instead, it is more beneficial to spread the scheduling time across all the slots within the frame.
	In addition, we develop a low-complexity algorithm providing a solution in polynomial time and it guarantees that the timely throughput constraint is satisfied. Simulation results show that the solution of the low-complexity is close to the near-optimal algorithm for certain values of transition probabilities of the channels.

	\section{System Model}
	We consider two users transmitting their information in \them{the} form of packets to a single receiver over a wireless fading channel, as shown in Fig. \ref{fig:systemmodel}. 
	\begin{figure}[h!]
		\centering
		\includegraphics[scale=0.2]{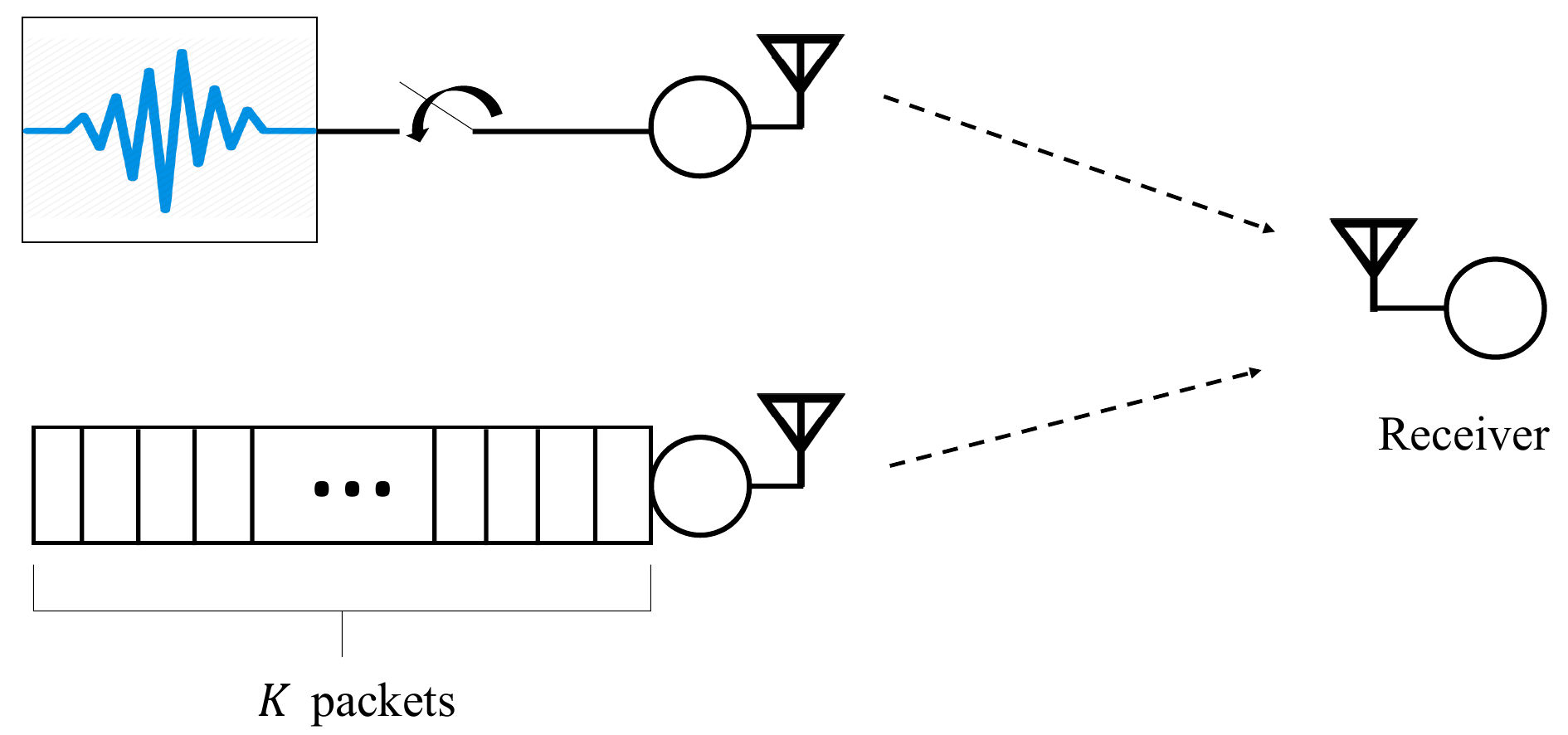}
		\caption{\them{System with an \ac{AoI}-oriented user and a deadline-constrained user}.}
		\label{fig:systemmodel}
	\end{figure}
	Let $i \in \{1,2\}$ denote the $i^{\text{th}}$ user of the system. Time is assumed to be slotted, and let $t \in \mathbb{Z}_{\geq 0}$ denote the $t^{\text{th}}$ slot, where $t \in \mathbb{Z}_{\geq 0}$ is the set of non-negative integer numbers. We consider a centralized scheduler that at every slot  decides to schedule up to one user. Let $u_{i}(t)$ denote the decision of the scheduler, where
	\begin{align}
		u_{i}(t) =
		\begin{cases}
			1\text{ , if user } i \text{ is scheduled at  time slot } t\text{,}\\
			0\text{, otherwise,}
		\end{cases}
	\end{align}
	and \them{$\bm{u}(t) = \left[u_1(t)\text{ } u_2(t)\right]^T$.}
	Note that $\sum\limits_{i} u_i(t) \leq 1$, $\forall t$.
	Due to the \them{wireless} nature of the channels, we assume that a packet is successfully transmitted from user $i$ to the receiver with some probability. 
	Let $d_{i}(t)$ denote the successful packet reception of user $i$, given that $u_i(t)=1$, where
	\begin{align}
		d_i(t) = 
		\begin{cases}
			1\text{, successful packet reception for user }i\text{,}\\
			0\text{, otherwise,}
		\end{cases}
	\end{align}
	and \them{$\bm{d}(t) = \left[d_1(t) \text{ }d_2(t)\right]^T$.}
	
	User $1$ is an \ac{AoI}-oriented user who either samples and transmits fresh information to the receiver or remains silent depending on the scheduling policy. Let $A(t) \in \mathbb{Z}_{> 0}$ represent the \ac{AoI} of user $1$ at the receiver. We assume that the value of the \ac{AoI} is bounded by  $A_{\text{max}} >0$. This assumption is considered for the following two reasons:
	1) in practical applications, values of \ac{AoI} that are larger than a threshold will not provide additional information about the staleness of the packet, \cite{ceran2019reinforcement,zhou2019joint,abd2020aoi};
	2) assuming unbounded \ac{AoI} will complicate significantly the solution of the optimization problem without giving us additional insights into the performance of the system. 
	Moreover, this assumption has been widely used in recent works that study the average AoI, \cite{zhou2019joint, moltafet2019power, ceran2019reinforcement}.
	
	The evolution of the \ac{AoI} at the receiver is described as
	\begin{align}
		A(t+1) = \begin{cases}
			1\text{, successful packet transmission of user } 1\text{,}\\
			\min{\{A_{\text{max}}, A(t)+1\}}\text{, otherwise.}
		\end{cases}
		\label{eq:AgeEvolution}
	\end{align}
	The time average \ac{AoI} is defined as
	$\bar{A} = \lim\limits_{t\rightarrow \infty} \sup  \frac{1}{t} \sum\limits_{\tau=0}^t \mathbb{E}\{A(\tau)\}\text{,}$
	where the expectation is with respect to the scheduling policy and the channel randomness. Note that we use \them{a} generate-at-will policy.
	Furthermore, since we do not consider sampling cost, user $1$ does not retransmit a packet if the transmission fails. Instead, user $1$ samples new information whenever it is scheduled in one of the following slots. We consider that the sampling and transmission process needs one time slot to be performed. 
	
	User $2$ is deadline-constrained, and it includes packets that must be transmitted within a specific time frame, i.e., before a deadline. More specifically, we consider that $K$ packets arrive in the queue of user $2$ every $T$ slots, where $K\leq T$. We consider that a packet needs one slot to be transmitted. The time between two adjacent packet arrivals is a time frame whose length is $T$ time slots. Let $m \in \mathbb{Z}_{\geq 0}$ denote the $m^{th}$ frame, and $t_m=mT$ be the first slot of frame $m$.
	We also denote by $f_{m}(t)$ the \them{time interval} between the beginning of the current frame and  slot $t$, i.e., $f_{m}(t) = t - mT$, $t\in\{mT, mT+1, \ldots, mT+T-1\}$. Let $Q(t)$ denote the number of packets that are in the queue of user $2$ in  time slot $t$. The evolution of the queue is described as
	$Q(t+1) = \max\{Q(t)-d_2(t),0\}\bm{1}_{\{f_m(t) \neq 0\}} + \them{K}\mathbf{1}_{\{f_m(t)  = 0\}}\text{, } \forall t \in\{mT, \ldots, mT+T-1\}\text{.}$
	
	\subsection{Frame-based Timely Throughput Requirements}
	The  \textit{timely throughput} measures the  average number of successful deliveries, i.e., the packets delivered before the deadline \cite{lashgari2013timely,TheoryQoS2009}. 
	Timely-throughput was proposed in \cite{TheoryQoS2009} as an analytical metric for evaluating both throughput and Quality of Service for deadline-constrained traffic.
	In this work, we are interested \them{in keeping}  \textit{frame-based  timely throughput} \them{above a threshold}, for user $2$. \them{The frame-based timely throughput} is defined as
	\begin{align}
		\lim\limits_{M\rightarrow \infty} \frac{1}{M+1} \sum\limits_{m=0}^M\left(\sum\limits_{\tau = mT}^{(m+1)T -1} \mathbb{E} \left\{d_2(\tau) \right\} \right)\text{.}
		\label{def:timelythroughput}
	\end{align}
	Our motivation for defining  \textit{per frame average timely throughput} is that it is not only important to serve as many packets as possible, but also to keep high QoS for every frame. \them{In other words, to ensure a high average number of successful packets delivery before their deadline expiration, i.e., before the end of the frame.} Consider a real-time video transmission where videos consist of frames. We care to have high-quality video transmission, and also a smooth transmission of the video from frame to frame which can be captured by \them{ensuring that the value of the expression in} \eqref{def:timelythroughput} \them{is above a threshold}.
	
	\subsection{Channel Model}
	In this work, we consider two cases for the channel model: \them{1)} i.i.d channel over slots, \them{2)} correlated channels over slots. 
	
	\subsubsection{i.i.d channels} In the i.i.d case, the channels of user $1$ and user $2$ are considered as Bernoulli processes with success probabilities $p_1$ and $p_2$, respectively. Therefore, if we schedule user $i$, the probability of the successful transmission is $p_i$. The success or failure of a transmission does not depend on the channel state of the previous slot.
	\subsubsection{Correlated channels}
	The channel of each user $i$ is assumed to be a time-correlated fading channel and each one evolves as a two-state Gilbert-Elliot model. The evolution of the channel states can be modeled as a two-state Markov chain. Let $h_i(t)$ denote the channel state of user $i$ at  time slot $t$, which is modeled as a Markov chain with two states and $\mathbf{h}(t) = \left[h_1(t)\text{ }h_2(t)\right]^T$. ``Bad" state represents deep fading of the channel and any transmission will fail. ``Good" state represents mild fading of the channel and any transmission will succeed. The channel transition probabilities are given by $\Pr\{h_{i}(t+1)=1|h_i(t)=1\} = p_{11,i}$, $\Pr\{h_{i}(t+1)=1|h_i(t)=0\} = p_{01,i}$. We consider delayed channel sensing for both users i.e., the channel state for each user $i$ is known at the receiver only at the end of each slot. Therefore, at the beginning of each slot, the scheduler knows the channel state based on the previous \ac{CSI} and with some probability.

	\section{Problem Formulation}
	\revtwo{In this section, we provide the formulation of the problem, and we give some basic definitions and assumptions that are necessary for designing the scheduling policies and the corresponding performance guarantees.}
	\begin{defi}[Scheduling Policy]
		A scheduling policy $\omega$ is a (possibly randomized) rule of scheduling user $i$ at each time slot $t$.
		Policy $\omega$ takes into account only the information of the state of the system at time slot $t$, which includes, the information about the channel state (i.e., success probability), the queue length of user $2$, the value of the \ac{AoI} of user $1$.
		Note that since we consider random channels, the outcome of a scheduling decision is a random variable as well. A policy $\omega$ specifies a probability distribution $\mathbf{u}^{\omega}(t)$, where $u_i^{\omega}(t)$ is the probability of scheduling user $i$ at time slot $t$ given the system state. The set of such policies is denoted by $\Omega$.
	\end{defi}
	In this work, our target is to find a policy $\omega$ that solves the following optimization problem
	\begin{subequations}
		\label{optproblem}
		\begin{align}
			\min\limits_\omega \quad & \bar{A}^{\omega}\label{opt:obj} \\
			\text{s.~t.} \quad &  	\lim\limits_{M\rightarrow \infty} \inf \frac{1}{M+1} \sum\limits_{m=0}^M\left(\sum\limits_{\tau = mT}^{(m+1)T-1} \mathbb{E} \left\{d^{\omega}_2(\tau) \right\} \right) \geq q \label{opt:constr}\text{,}
		\end{align}
	\end{subequations}
	where $q$ is the minimum per frame average timely throughput requirements of user $2$, $0\leq q\leq K$.
	
	
	\textbf{Remark 1.}
	\textit{The optimization problem in \eqref{optproblem} belongs to the general class of \acp{CMDP}, which are usually difficult to solve \cite{altman1999constrained}. In \cite{salodkar2008line, djonin2007q}, authors utilize the Lagrange multiplier approach to treat power minimization for single-queue systems with an average delay constraint. Our approach that is described below, while similar to the Lagrange multiplier approach, can handle the stochastic nature of the shortest path problem over frames, which is otherwise hard to achieve. This specific characteristic of our problem renders the Lagrange multiplier approach inappropriate. For this reason, we are deploying virtual queues which are able to capture the dynamics of the system under consideration. Additionally, the use of virtual queues provides information about the stability of the overall system. 
		In more detail, function $d_2(\tau)$ in constraint (5b) is a random variable that depends on the channel, and the transmission decisions at every time slot. The constraint represents the average value of served packets per time frame. The specific characteristic of this problem is that it lies in the category of stochastic shortest path problems. The parameters of the stochastic path change from frame to frame. Therefore, we have to find a way to monitor these changes in order to satisfy the time average constraint. In other words, we have to dynamically update a weight from frame to frame to monitor how much of the constraint has been satisfied. That is why we choose to solve this problem by incorporating virtual queues that capture the dynamics of the system. Classical Lagrange methodology cannot be directly applied because the Lagrangian multipliers do not capture the stochastic result of function $d_2(\tau)$.}

	\begin{defi}[Feasible Region of Timely-Throughput Requirements]
		Consider the set of all policies $\Omega$, and denote by $\bar{q}^\omega$ the timely-throughput that is achieved by applying policy $\omega$. Then, the feasible region of timely throughput requirements, denoted by $\Gamma$, is the set of variables $q$ which satisfies
		$\Gamma  =  \cup_{\omega \in \Omega} \left\{ q \in [0, K]  | q\leq \bar{q}^\omega \right\}\text{.}$
	\end{defi}

	\subsection{Slackness Assumptions}
	The problem in \eqref{optproblem} is a \ac{CMDP} with state $s(t) = (A(t),Q(t),\mathbf{h}(t))$, where $\mathbf{h}(t) = [h_1(t)\text{ }h_2(t)]^T$, and $h_i(t)$ is the channel information of user $i$ at time slot $t$. Under mild assumptions (such as the state space being finite, and the action space is also finite) the MDP has an optimal stationary policy that chooses $\mathbf{u}(t)$ as a stationary and possibly randomized function of the state $s(t)$ only \cite{neely2013dynamic}. Note that the system experiences regular renewals, i.e., at the beginning of each frame a batch of $K$ packets arrive at the queue while the remaining packets from the previous frame has been discarded. Therefore, the performance of any $s(t)$-only policy can be characterized by ratios of expectations over one renewal time \cite{NeelyBook}. Thus, we make the following assumption.
	
	\textit{Assumption 1}: There exists a policy $\omega_1\in\Omega$ that satisfies the following, over any renewal frame:
	\begin{align}\label{eq:Ass1Obj}
		\frac{\mathbb{E}\left[\sum\limits_{\tau=mT}^{(m+1)T-1}  A^{\omega_1}(\tau) \right]}{T} & = \bar{A}_{\text{opt}}\text{, }\\
		\them{q}- \mathbb{E}\left[ \sum\limits_{\tau=mT}^{(m+1)T-1}  d_2^{\omega_1}(\tau)\right] & \leq 0\text{,}\label{eq:Ass2Obj}
	\end{align}
	where $A^{\omega_1}(\tau)$, and $d_2^{\omega_1}(\tau)$, are the values of $A(t)$ and $d_2(t)$ obtained by applying policy $\omega_1$, and $\bar{A}_{\text{opt}}$ is the optimal value of the time average AoI.
	Note that Assumption $1$ is mild and holds whenever problem $\eqref{optproblem}$ is feasible, i.e., $\forall q \in \Gamma$. We now make a stronger assumption guaranteeing that the constraint in \eqref{opt:constr} is met with $\epsilon$-slackness. In the following assumption, we focus only on the satisfaction of the constraints.
	This assumption is related to standard ``Slater-type" assumptions in optimization theory \cite{bertsekas1997nonlinear}.
	
	\textit{Assumption 2:} There is a value $\epsilon>0$ and a policy $\omega_2$ that satisfies the following over any renewal frame:
	\begin{align}
		\label{eq:Ass2const}
		\frac{q -\mathbb{E} \left[ \sum\limits_{\tau=mT}^{(m+1)T-1}  d_2^{\omega_2}(\tau) \right]}{T} \leq -\epsilon\text{,}
	\end{align}
	\them{where $d_2^{\omega_2}(\tau)$ is the value of $d_2(\tau)$ obtained by applying policy  $\omega_2$.}
	In the next section, we describe our proposed dynamic control algorithm. In Theorem 2, we prove that if there exists a policy that satisfies Assumption 2, the virtual queue is bounded and therefore, it is strongly stable. As a consequence, the average timely throughput constraint is satisfied.
	
	\section{Dynamic Control Algorithm}
	Before describing our proposed algorithm that solves $\eqref{optproblem}$, let us recall a definition and a basic theory that comes from the theory of stochastic processes \cite{meyn2012markov}. Consider a system with \them{$N$} queues. The number of unfinished jobs of queue \them{$n$} is denoted by $q_n(t)$, and \them{$\mathbf{q}(t) = \left[q_1(t)\text{ }\ldots\text{ }q_N(t)\right]^T$}. The Lyapunov function and the Lyapunov drift are denoted by $L(\mathbf{q}(t))$ and $\Delta(L(\mathbf{q}(t))) \triangleq \mathbb{E}\left\{L(\mathbf{q}(t+1)) -  L(\mathbf{q}(t)) | \mathbf{q}(t)\right\}$, respectively, and they are defined below.
	\begin{defi}[Lyapunov function]
		A function $L: \mathbb{R}^N \rightarrow \mathbb{R}$ is a Lyapunov function if it has the following properties:
		1) $L(\mathbf{x}) \geq 0$, $\forall \mathbf{x} \in \mathbb{R}^{\them{N}}$,
		2) It is non-decreasing in any of its arguments,
		3) $L(\mathbf{x}) \rightarrow \infty$, as $||\mathbf{x}|| \rightarrow +\infty$\text{.}
	\end{defi}
	\begin{defi} [Strong Stabilitity]
		A discrete time process $\mathbf{Q}(t)$ is \textit{strongly strable} if:
		
		$\lim\sup\limits_{t\rightarrow \infty } \frac{1}{t} \sum\limits_{\tau = 0}^{t-1} \mathbb{E}\left\{|\mathbf{Q}(\tau)|\right\} \leq \infty \text{.}
		$
	\end{defi}
	\begin{thm}[Lyapunov Drift]
		If there exists positive values $B$ and $\varepsilon>0$ such that  for all time slots we have $\Delta(\mathbf{q}(t)) \leq B - \varepsilon \sum\limits_{n=1}^N q_n (t)$, then the system $\mathbf{q}(t)$ is strongly stable.
	\end{thm}
	The intuition behind Theorem 1 is that if we have a queueing system and we provide an algorithm for which the Lyapunov drift \them{becomes negative for large queue sizes, then the Lyapunov function decreases and subsequently the queue sizes. As a consequence, the queues remain bounded and the overall system is stable.}
	
	We define a virtual queue $Z(t)$ to represent constraint  $\eqref{opt:constr}$, where $Z(0)=0$. We update the value of the virtual queue as 
	$Z(t+1) = \max \left[ Z(t)-d_2(t),0 \right] + \frac{q}{T}\text{.}$
	Process $Z(t)$ can be seen as a queue with ``service rate" $\bar{d}_2$ and ``arrival rate" $\frac{q}{T}$. We will show that the average constraint in $\eqref{opt:constr}$ is transformed into a queue stability problem. 
\begin{defi}
		A discrete time process $Q(t)$ is rate stable if $\lim\limits_{t\rightarrow \infty} \frac{Q(t)}{t} =0$ with probability $1$.
\end{defi}
\begin{lemma}
	If $Z(t)$ is rate stable, then constraint  $\eqref{opt:constr}$ is satisfied. 
\end{lemma}
\begin{proof}
	Using the basic sample property  \cite[Lemma 2.1, Chapter 2]{NeelyBook}, we have
	\begin{align}
		\frac{Z(t_M)}{t_M} - \frac{Z(0)}{t_M} \geq \frac{1}{M+1} \sum\limits_{m=0}^{M} \left(\sum\limits_{\tau=mT}^{(m+1)T -1} d_2(\tau) -\frac{q}{T}\right) \text{.}
	\end{align}
	Therefore, if $Z(t_M)$ is rate stable, so that $\frac{Z(t_M)}{t_M} \rightarrow 0$, with probability $1$, then the constraint in $\eqref{opt:constr}$ is satisfied.
\end{proof}

\subsection{Lyapunov Drift}
We define the following quadratic Lyapunov function as
$
L(Z(t)) \triangleq \frac{1}{2} Z^2(t)\text{.}
$
We define the \textit{frame-based Lyapunov drift} as 
$\Delta(Z(t_m)) \triangleq \mathbb{E} \left[ L(Z(t_m+T)) - L(Z(t_m)) | Z(t_m) \right]\text{,}$
where $t_m=mT$ is the starting slot of the $m^{\text{th}}$ frame. 
\begin{lemma}
	Under any policy $\mathbf{u}(\tau)$ for all slots during a renewal frame $\tau \in \left\{t_m, \ldots, t_m+T-1\right\}$, we have $
	\Delta(Z(t_m))   \leq B + \mathbb{E}\left[ G(t_m) | Z(t_m) \right]\text{,}
	$
	where $G(t_m)$ is defined as
	$
	G(t_m) \triangleq Z(t_m) \sum\limits_{\tau=t_m}^{t_m+T-1} (q - d_2(\tau)) \text{,}
	$
	and $B$ is a finite constant defined as 
	$B \triangleq \frac{T q^2+T(T-1)}{2} \text{.}$
\end{lemma}
\begin{proof}
	See Appendix \ref{Appendix:Proof of Lemma 2}.
\end{proof}

\subsection{Frame-Based Drift-Plus Penalty Algorithm}
In order to provide a solution to the optimization problem in $\eqref{optproblem}$, we implement a policy over the course of the frame to minimize the following expression
\begin{align}
	\label{opt:penaltyalgo}
	\min\limits_{\mathbf{u}(t)} \mathbb{E} \left[G(t_m) + V\sum\limits_{\tau=t_m}^{t_m+T-1} A(\tau) | Z(t_m)\right]\text{,}
\end{align}
where the expectation is with respect to the policy and the randomness of the channel. 
The problem in \eqref{opt:penaltyalgo} is a \textit{stochastic shortest path problem} which usually is solved approximately \cite{neely2013dynamic}.

The Stochastic Shortest Path (SSP) problem is a Markov Decision Process (MDP) that generalizes the classic shortest path problem. In an SSP, we consider a graph with nodes $1, 2, \ldots, n, t$. Each node belongs to a stage. At each node, we must select a probability distribution over all possible successor nodes $j$ out of a given set of probability distributions $p_{ij}(u)$ parametrized by a control $u \in U(i)$. Our goal is to traverse the graph with the minimum cost and reach the destination $t$, \cite{bertsekas2011dynamic}. In this work, our goal is to minimize (10) over a frame with time horizon $T$ slots. Within a frame, the problem is considered as an SSP with $T$ stages. Every stage contains a set of nodes. At every node, i.e., at every state, we need to take decisions out of a given set of probability distributions (in (16), (17)). The final goal is to reach the $T$-th slot (or $T$-th stage) with the minimum cost based in (10) that corresponds to the minimum length in an SSP.

In the next subsection, we analyze the performance of the algorithm under the assumption that we have a policy that can approximate \eqref{opt:penaltyalgo}.

\subsection{Approximation Theorem}
\them{\textit{Assumption 3:}}
For constants, $C\geq 0$, $\delta \geq 0$, define a $(C, \delta)$ - approximation of  \eqref{opt:penaltyalgo} to be a policy for choosing $\mathbf{u}(t)$ over a frame $\tau \in \left\{t_m, \ldots, t_m+T-1\right\}$ such that 
\begin{align}\nonumber
	&\mathbb{E} \left[ G(t_m)  + V\sum\limits_{\tau=t_m}^{t_m + T -1} A(\tau) | Z(t_m)   \right] \leq \\	\nonumber
	&\mathbb{E} \left[G^\text{opt}(t_m) + V\sum\limits_{\tau=t_m}^{t_m+T-1} A_{\text{opt}}(\tau) | Z(t_m)\right] \\ \label{ineq: CdApprox} &+ C + \delta Z(t_m) +V\delta \text{,}
\end{align}
\them{where $A_\text{opt}$ and $G_\text{opt}$ are the optimal values.}

\revtwo{
	\textbf{Remark 2.}
	\textit{
		Assumption $3$ is necessary for proving the performance bounds of our proposed algorithm in the following theorem. Note that by solving optimally the \textit{stochastic shortest path problem} in \eqref{opt:penaltyalgo}, the values of $V,\delta$ are $0$. Assumption $3$ is used in this paper to describe approximation algorithms in general.}}
\begin{thm}\label{Theorem: Performance and Constraints}
	Suppose that Assumptions 1, 2, hold for a given $\epsilon>0$, and suppose we use a $(C, \delta)$-approximation every frame so that Assumption 3 holds. If $\epsilon>\frac{\delta}{T}$, then constraint \eqref{opt:constr} is satisfied and 
	\begin{align}
		\label{boundQ}
		\lim\limits_{R\rightarrow \infty} \sup \frac{1}{R} \sum\limits_{r=0}^R \mathbb{E} \left[Z(t_m)\right] \leq \frac{B +C +V(T\bar{A}_{\text{opt}} + \delta)}{\epsilon T - \delta}
	\end{align}
	and 
	\begin{align}\nonumber
		\lim\limits_{t\rightarrow \infty} \sup \frac{1}{t} \sum\limits_{\tau=0}^{t-1} \mathbb{E} \left[A(\tau)\right] &\leq 
		\frac{B}{VT} + \gamma (A_{\max} - 1 )
		 \\\label{boundAge} &+(1-\gamma)\them{A_{\text{opt}}}  
		+ \frac{C}{VT} + \frac{\delta}{T}\text{,}
	\end{align}
	where $\gamma = \frac{\delta}{\epsilon T}\text{.}$
\end{thm}

\begin{proof}
	See Appendix \ref{Appendix: Proof of Theorem 2}.
\end{proof}
\revtwo{\textbf{Remark 3.} \textit{Inequality \eqref{boundQ}, in Theorem \ref{Theorem: Performance and Constraints}, indicates that by minimizing the upper bound of the Lyapunov drift-plus-penalty expression, we can guarantee that the virtual queue is strongly stable and therefore, the timely-throughput requirements are satisfied. In addition, inequality \eqref{boundAge} indicates that the time average AoI is bounded by an expression. As it is shown, by parametrizing the value of $V$ we can take insights on how close to the optimal  solution the $(C,\delta)$ approximation is. In addition, we obtain that if we apply an optimal algorithm for solving problem \eqref{opt:penaltyalgo}, we get an exact solution, and therefore, $C=\delta=0$. From Theorem \ref{Theorem: Performance and Constraints}, we get 
		\begin{align}\label{ineq:virtualqueue}
			\lim\limits_{R\rightarrow \infty} \sup \frac{1}{R} \sum\limits_{r=0}^R \mathbb{E} \left[Z(t_m)\right] & \leq \frac{B +V(T\bar{A}_{\text{opt}})}{\epsilon T}\\\label{ineq:boundAoIOptimality}
			\lim\limits_{t\rightarrow \infty}\sup \frac{1}{t} \sum\limits_{\tau =0}^{t-1}\mathbb{E}[A(\tau)]
			& \leq \frac{B}{VT} + A_\text{opt}\text{.}
		\end{align}
		We observe from the expressions above that, as we increase $V$, we get a solution that is closer to the optimal one because fraction $\frac{B}{VT}$ in \eqref{ineq:boundAoIOptimality} becomes smaller as $V$ increases. However, we get a larger bound in the virtual queue $Z(t_m)$, as shown in \eqref{ineq:virtualqueue}, and therefore, a slower convergence regarding the timely-throughput constraints as shown also in the results. Thus, by changing the value of $V$, we obtain a trade-off between the performance of the algorithm between the value of average AoI and the convergence regarding the timely-throughput constraints.
}}

\section{Solution of the MDP}
The problem in \eqref{opt:penaltyalgo} is an \ac{MDP} problem. Let $\mathcal{A} = \left\{1, 2, \ldots, A_{\text{max}}\right\}$ denote the set of  possible values of \ac{AoI} of user $1$. Furthermore, let $\mathcal{Q} = \left\{0,1,2,\ldots, L\right\}$ be the set of possible values of the queue of user $2$. Then, $\mathcal{A}(t)$ $\in$ $\mathcal{A}$, and $Q(t)$ $\in$ $\mathcal{Q}$. A transmission policy $\mathbf{u}(t)$ specifies the decision rules every time slot $t$. Note that the described \ac{MDP} problem is a finite-horizon problem. We solve the optimization problem at every frame. At the beginning of each frame, we know the channel conditions of the previous slot for each users, the state of the queue (it is always $L$ packets at the beginning of the frame), and the value of the \ac{AoI} of user $1$. 

The next state depends on both the scheduler's decision and the channel states. Note that we schedule user $2$ only if it has remaining packets in its queue, and recall that at the end of frame $m$, we drop all the remaining packets, if there is any. 

\subsection{Transition Probabilities}
The transition probabilities of the system's states depend on the channel model. Below we describe the transition probabilities for each channel model, and the corresponding system states.
\subsubsection{i.i.d channels} In this case, the system state in time slot $t$ is described by  $s(t) = (A(t),Q(t))$. The transition probabilities for the i.i.d channel case are described in \eqref{eq:TransProbiidCase}.
\begin{figure*}[t!]
	\begin{small}
		\begin{align}\label{eq:TransProbiidCase}
			P_{s_t\rightarrow s_{t+1}}=
			\begin{cases}
				p_1\text{, if } \mathbf{u}(t) = \left[1\text{ }0\right]^T\text{ and } s(t+1) = (1,Q(t))\text{,}\\
				1-p_1 \text{, if } \mathbf{u}(t) = \left[1\text{ }0\right]^T\text{ and } s(t+1) = (\min\left\{A(t)+1, A_{\text{max}}\right\},Q(t))\text{,}\\
				p_2\text{, if } \mathbf{u}(t) = \left[0\text{ }1\right]^T\text{ and } s(t+1) = (\min\left\{A(t)+1, A_{\text{max}}\right\}, Q(t)-1)\text{,}\\
				1-p_2\text{, if } \mathbf{u}(t) = \left[0\text{ }1\right]^T \text{ and } s(t+1) = (\min\left\{A(t)+1, A_{\text{max}}\right\}, Q(t))\text{.}
			\end{cases}
		\end{align}
	\end{small}
\end{figure*}

\subsubsection{Time-correlated channels} In this case the system state in time slot $t$ is described by $s(t) = (A(t),Q(t),\mathbf{h}(t))$. The transition probabilities for the case of Gilbert-Elliot channel model  are described in \eqref{eq:TransProbTimeCorrelatedCase}.
\begin{figure*}[t!]
	\begin{small}
		\begin{align}\label{eq:TransProbTimeCorrelatedCase}
			P_{s_t\rightarrow s_{t+1}}=
			\begin{cases}
				p_{11,1}, \text{if } \mathbf{u}(t) = \left[1\text{ }0\right]^T\text{, } h_1(t)=1\text{, and } s(t+1) = (1,Q(t),(1,x)),\\
				1- p_{11,1} \text{, if } \mathbf{u}(t)=\left[1\text{ }0\right]^T\text{, } h_1(t) = 1\text{, and } s(t+1) = (\min\left\{A(t)+1,A_{\text{max}}\right\},Q(t), (0,x)),\\
				p_{01,1}\text{, if } \mathbf{u}(t)=\left[1\text{ }0\right]^T\text{, } h_1(t)=0\text{, and } s(t+1) = (1,Q(t),(1,x)), \\
				1-p_{01,1} \text{, if } \mathbf{u}(t)=\left[1\text{ }0\right]^T\text{, } h_1(t) = 0\text{, and } s(t+1) = (\min\left\{A(t)+1,A_{\text{max}}\right\},Q(t), (0,x)),\\
				p_{11,2}, \text{if } \mathbf{u}(t) = \left[0\text{ }1\right]^T\text{, } h_2(t)=1\text{, and } s(t+1) = (\min\left\{A(t)+1,A_{\text{max}}\right\},Q(t)-1,(x,1)),\\
				1- p_{11,2} \text{, if } \mathbf{u}(t)= \left[0\text{ }1\right]^T\text{, } h_2(t) = 1\text{, and } s(t+1) = (\min\left\{A(t)+1,A_{\text{max}}\right\},Q(t), (x,0)),\\
				p_{01,2}\text{, if } \mathbf{u}(t)= \left[0\text{ }1\right]^T\text{, } h_2(t)=0\text{, and } s(t+1) = (\min\left\{A(t)+1,A_{\text{max}}\right\},Q(t)-1,(x,1)), \\
				1-p_{01,2} \text{, if } \mathbf{u}(t)=\left[0 \text{ }1\right]^T\text{, } h_2(t) = 0\text{, and } s(t+1) = (\min\left\{A(t)+1,A_{\text{max}}\right\},Q(t), (x,0)),
			\end{cases}
		\end{align}
	\end{small}
\end{figure*}
where $x$ is used to show that the value of the corresponding element does not affect the state transition.
\subsection{Backward dynamic programming algorithm}
Initially, we drop the frame indices and take $t$ $\in \left\{0,1,\ldots, T-1\right\}$.
As a first step, we consider that the transmission error probabilities are fixed, i.e., the channels are i.i.d over the slots. 
In our system model, we take an action in time slot $t$, and we observe the cost  in time slot $t+1$. This happens because of the randomness of the channels for both i.i.d and correlated channels. If we transmit a packet, it will successfully be transmitted with some probability. We know whether the transmission is successful or not at the end of the slot due to ACK/NACK.
Below we define the costs for the different channel models.
\subsubsection{i.i.d channels}
We denote the cost received in slot $t+1$ given the state $s(t)$, the decision $\mathbf{u}(t)$, and the information received $W_{t+1}$ (successful transmission or failure), as  $\hat{C}_{t+1}(s(t+1), W_{t+1})\text{.}$ Then, the instantaneous cost at time slot $t$ is described as 
\begin{align}\nonumber
	&C_t(s(t),\mathbf{u}(t))  = \mathbb{E} \left\{\hat{C}_{t+1}(s(t+1), W_{t+1})|s(t), \mathbf{u}(t)\right\}\\& = 
	\begin{cases}
		Z q  + V (p_1 + (1-p_1)\min(A(t)+1,A_{\text{max}}))\text{, if } u_1(t)=1\text{,}\\
		Z(q-p_2)  + V(\min(A(t)+1,A_{\text{max}})) \text{, if } u_2(t)=1\text{.}
	\end{cases}
	\label{eq:instantCost}
\end{align} 
The Bellman's equation is described below
\begin{align}\nonumber
	&V_t (s(t))  =
	\min\limits_{\mathbf{u}(t)} \left(C_t(s(t),\mathbf{u}(t)) \right.\\ \label{Bellman}
	&\left. +\gamma \sum\limits_{s'\in \mathcal{S}} \Pr(s(t+1)=s'|s(t),\mathbf{u}(t)) V_{t+1}(s')\right) \text{,} 
\end{align}
where $0<\gamma<1$. 
\begin{figure*}[t!]
	\begin{align}\nonumber
		&C_t(s(t),\mathbf{u}(t))  = \mathbb{E} \left\{ \hat{C}_{t+1}(s(t+1), W_{t+1} | s(t), \mathbf{u}(t)) \right\} \\ \label{eq: timeCorrelatedCost}& = 
		\begin{cases}
			Zq + V(p_{11,1} + (1-p_{11,1}) \min\left\{A(t)+1,A_{\text{max}}\right\})\text{, if } u_1(t) = 1\text{ and } h_1(t)=1\text{,}\\
			Zq + V(p_{01,1} + (1-p_{01,1}) \min\left\{A(t)+1,A_{\text{max}}\right\})\text{, if } u_1(t) = 1\text{ and } h_1(t)=0\text{,}\\
			Z(q-p_{11,2} )+ V( \min\left\{A(t)+1,A_{\text{max}}\right\})\text{, if } u_2(t) = 1\text{ and } h_2(t)=1\text{,}\\
			Z(q-p_{01,2} )+ V( \min\left\{A(t)+1,A_{\text{max}}\right\})\text{, if } u_2(t) = 1\text{ and } h_2(t)=0\text{,}
		\end{cases}
	\end{align}
\end{figure*}
\subsubsection{Time-correlated channels}
In this case, the instantaneous cost at time slot $t$ is described in \eqref{eq: timeCorrelatedCost},
where $\hat{C}_{t+1}(s(t),\mathbf{u}(t), W_{t+1} | s(t), \mathbf{u}(t), \mathbf{h}(t)$ is the cost received in time slot $t$. 
The Bellman's equation is described below
\begin{align}\nonumber
	&V_t (s(t)) 
	=\min\limits_{\mathbf{u}(t)} \left(C_t(s(t),\mathbf{u}(t)) \right. \\\label{BellmanCorr}
	&\left.+\gamma \sum\limits_{s'\in \mathcal{S}} \Pr(s(t+1)=s'|s(t),\mathbf{u}(t)) V_{t+1}(s')\right) \text{,} 
\end{align}
where $0<\gamma<1$. 

We can solve the recursions in \eqref{BellmanCorr} and \eqref{Bellman} by using \text{backward dynamic programming}. 
We denote by $\mathcal{S}$ the set with all possible states. The algorithm is shown below 
\begin{figure*}[t!]
\begin{center}
	\fbox{
		\begin{minipage}{40em}
			\textbf{Algorithm 2: } Backward Dynamic Programming
			
			\textbf{Step 0.} Initialization:
			
			Initialize the terminal contribution $V_T(s_T)$. Usually, we set the value of $0$ to $V_T(s_T)$.
			
			Set $t=T-1$.
			
			\textbf{Step 1.} Calculate: 
			\begin{align}\nonumber
				V_t (s(t))	=\min\limits_{\mathbf{u}(t)} \left(C_t(s(t),\mathbf{u}(t)) +\gamma \sum\limits_{s'\in \mathcal{S}} \Pr(s(t+1)=s'|s(t),\mathbf{u}(t)) V_{t+1}(s')\right) \text{,}
			\end{align}
			$\forall  s(t) \in \mathcal{S}$.
			
			\textbf{Step 2.} If $t>0$, decrement $t$ and return to step $1$. Else stop. 
		\end{minipage}
	}
\end{center}	
\end{figure*}

We can implement the backward dynamic programming described above for both channel cases. The idea of the algorithm is quite simple. The algorithm runs over the duration of each frame starting at the last slot, i.e., the $T^\text{th}$ of each frame. We initialize the value of being at each state at the last slot, and then, we calculate the value of each state at every time slot by going backward in time as shown in Step 1. 


\subsection{Discussion on the complexity of the proposed solution}
From the results of this paper, while not trivial, it is evident that we can extend
the work to a multi-user scenario based on the methodology provided in this work. More
specifically, in a multi-user scenario, the relaxation of the timely-throughput constraints can
be performed by applying the same methodology as the one utilized in the manuscript. That is, if we have $N$ deadline-constrained users, we can map each constraint to a virtual queue. Therefore, as a first step for the solution, the task will be to stabilize these virtual
queues in order to satisfy the constraints. The drift-plus-penalty algorithm can be applied in
order to solve the problem. The new part in a multi-user scenario will be the dimension of the
space of the MDP. In practice, the state space $\mathcal{S}$ for the problem 
becomes too large, as the number of nodes increases, to evaluate the value function $V_t(s_t)$ for all
states within a reasonable time. Exact solutions such as backward dynamic programming cannot be applied in practice, due to the widely
known curse of dimensionality \cite{powell2007approximate}. However, there are other approximation algorithms that
can solve large problems with low complexity. In Section VI, we provide a low-complexity algorithm that solves the problem in (5), with constraint satisfaction guarantees. 

\section{Low Complexity Algorithm}
	In this section, we provide a low-complexity algorithm that guarantees that the timely-throughput constraint is satisfied. The proposed algorithm takes into account one step ahead environment transitions, and it decides for every slot which user will transmit. To provide a low-complexity algorithm that takes into account also the statistics of the channel, we reformulate the optimization problem and define a new objective function as follows: 
	\begin{align}\label{obj:relaxed}
		\phi(t) = 
		\begin{cases}
			p + (1-p)\min (A(t)+1,A_\text{max}) \text{, if } u_1(t)=1\text{,}\\
			\min (A(t)+1,A_\text{max})\text{, if } u_2(t) = 1\text{,}
		\end{cases}
	\end{align}
	where $p=p_s$, in the case of i.i.d channels, and $p = p_{11}$, if $h(t-1)=1$, $p=p_{01}$, if $h(t-1)=0$, in the case of correlated channels.
	The time average of the objective function is defined as
	$\bar{f} = \lim\limits_{t\rightarrow \infty} \sup \frac{1}{t}\sum\limits_{\tau =0}^{t-1} \mathbb{E}[\phi(\tau)]\text{.}
	$
	The reformulated optimization problem is the following:
	\begin{subequations}
		\label{optreformulatedproblem}
		\begin{align}
			\min\limits_{\mathbf{u}(t)}~ \bar{\phi} \text{,}
			\text{ s.t.}~ \bar{d}_2 \geq q'\label{opt:constrref}\text{,}
		\end{align}
	\end{subequations}
	where $q' = q/T$.
	We define the slot-based Lyapunov drift as
	\begin{align}\label{drift:reformulated}
		\Delta(Z(t)) \triangleq \mathbb{E}[L(Z(t+1) - L(Z(t))|Z(t)]\text{.}
	\end{align}
	By utilizing the fact that $(\max [Q-b,0] + A)^2 \leq Q^2 + A^2 + b^2 + 2Q(A-b)$, we obtain an upper bound on the expression in \eqref{drift:reformulated} as following: 
	$
	\Delta(Z(t)) + V\mathbb{E}[\phi(t)] \leq B + \mathbb{E}[Z(t)(q'-d(t))]\text{,}
	$
	where $B=\frac{(q')^2+d^2_2(t)}{2} < \infty$. The \ac{psDPP} uses the value of the virtual queue, and the channel state, and decides the scheduling at every time slot according to the following optimization problem:
	\begin{align}\label{dpp:myopic}
		\min\limits_{\mathbf{u}(t)} V\phi(t) + Z(t)(q'-d_2(t))\text{.}
	\end{align}
The proposed low complexity algorithm (psDPP) seeks to minimize the objective function
in (25) at every time slot. The intuition behind that algorithm is that it greedily minimizes
the AoI while keeping Z(t) bounded so that the time average constraints are satisfied.
\begin{lemma}
\label{lemma:boundeddrift}
Consider a policy $\omega$, similar to Definition 1. A policy $\omega(t)$ takes probabilistic decisions independent of the state of the system, at every time slot $t$. Let $y(t) = q' - d(t)$. If problem \eqref{drift:reformulated} is strictly feasible, and the second moments of $y(t)$ and $\phi(t)$ are bounded, then there exists $\varepsilon>0$, for which there is an $\omega(t)$ policy such that the following holds
$\mathbb{E}[\phi^{*}(t)] \leq \epsilon\text{, } \mathbb{E}[\phi^{*}(t)] \leq \phi^\text{opt} + \varepsilon\text{,}$
where $y^{*}$ and $\phi^{*}$ are the resulting values of $\omega$ policy, and $\phi^\text{opt}$ is the optimal value function achievable by any optimal stationary randomized policy.
\end{lemma}
\begin{proof}
	The proof follows from \cite[Theorem 4.5]{NeelyBook}.
\end{proof}
\begin{thm}
	The \ac{psDPP} algorithm satisfies any feasible set of timely throughput constraints.
\end{thm}
\begin{proof}
	Since the per slot \ac{DPP} algorithm seeks to minimize expression \eqref{dpp:myopic}, we obtain that
	\begin{align}
		\Delta(Z(t)) + V\mathbb{E}[\phi(t)] & \leq B + Z(t)\mathbb{E}[{y_\text{DPP}(t)}] + V \mathbb{E}[\phi_\text{DPP}(t)]\\
		& \leq B + Z(t)\mathbb{E}(\phi^{*}(t)) + V\mathbb{E}[\phi^{*}(t) ]\text{.}
	\end{align}
	By considering the bound in Lemma \ref{lemma:boundeddrift}, we get
	$\Delta(Z(t)) + V\mathbb{E}(\phi(t)) \leq B + \epsilon Z(t) + V(\phi^{\text{opt}} +\epsilon)\text{,}$
	and taking $\epsilon \rightarrow 0$, we obtain
	$\Delta(Z(t)) + V\mathbb{E}(\phi(t)) \leq B+ V(\phi^{\text{opt}})\text{.}$
	The above expression is in the exact form of the Lyapunov optimization theory, \cite{NeelyBook}[Theorem 4.2]. Therefore, the virtual queue is mean rate stable, and the timely throughput constraint is satisfied. 
\end{proof}

\section{Complexity Analysis}
In this section, we provide the complexity analysis for the two proposed algorithms.
It is easy to see that the psDPP algorithm searches for the minimum value of (25) by iterating over all users. Thus, the complexity of the algorithm depends on the number of users and it is $\mathcal{O}(N)$, where $N$ is the number of total users in the system.

In Algorithm 2 (backward dynamic programming), the algorithm estimates the value function for each state at every slot, by proceeding backward in time within a frame. Thus, for every value of $u(t)$, i.e., for each scheduling decision, the algorithm calculates the expression in (21), given that we know the value of the value function for every possible $s'$ state. Let $N_\text{A}$ and $N_\text{D}$ be the number of AoI-oriented users and the number of deadline-constrained users, respectively. Then, the total number of states, in the case of a Gilbert-Elliot channel model, is $(A^\textbf{max})^{N_\text{A}} K^{N_{d}} 2^{N}$, where $N$ is the total number of users in the system. The backward dynamic programming provides the optimal solution to the problem in (10), but it suffers from the curse of dimensionality, specifically in the number of states.

In a nutshell, the psDPP-based has a much lower complexity than pfDPP. Specifically, its complexity is polynomial with respect to the number of users. On the other hand, pfDPP suffers from the curse of dimensionality. Nevertheless, pfDPP provides a performance bound for this system.

\section{Simulation Results}

\begin{figure}[h!]
\centering
\begin{subfigure}{0.49\textwidth}
\centering
\includegraphics[scale=0.35]{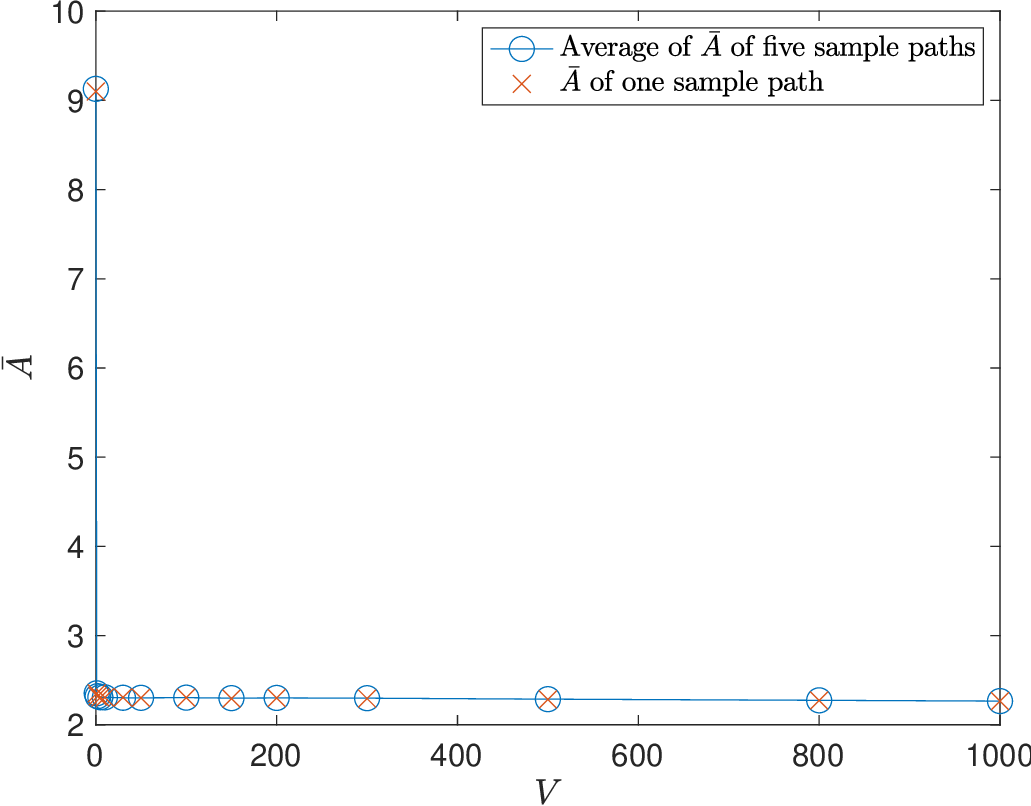}
\caption{Average \ac{AoI} for different values of  $V$.}
\label{fig:ageallvge}
\end{subfigure}
\vspace{1mm}
\begin{subfigure}{0.49\textwidth}
\centering
\includegraphics[scale=0.35]{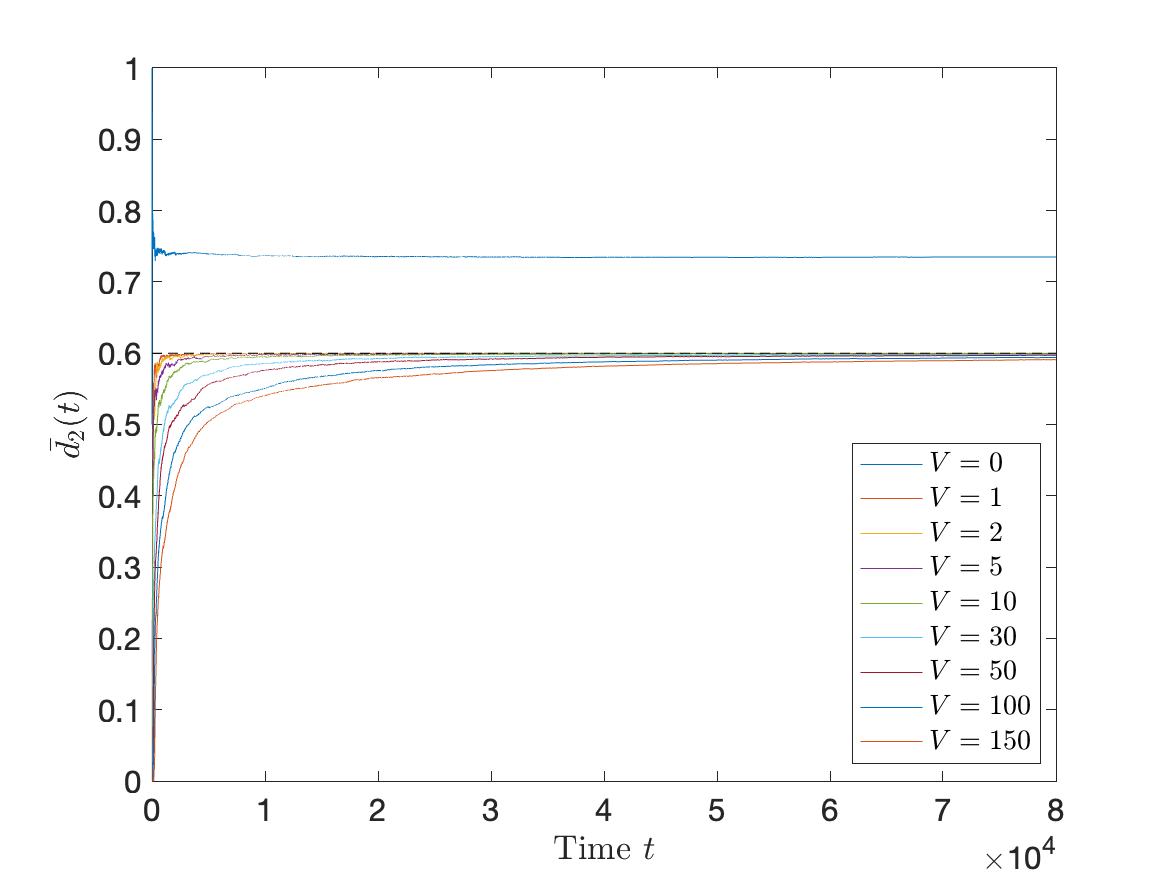}
\centering
\caption{Convergence of the timely throughput constraints.}
\label{fig:thrconvlongge}
\end{subfigure}
\caption{Performance analysis for average AoI and convergence of the algorithm.}
\label{Fig:AoIThroughput}
\end{figure}

\subsection{Performance Analysis for the near-optimal Algorithm}
In this section, we provide results to study the performance of the \ac{pfDPP}, which is proved to be near-optimal, in terms of the average value of the \ac{AoI} and the convergence regarding the timely throughput requirements. We investigate how different values of the weight factor $V$ can affect both the value of \ac{AoI} and the convergence of the algorithm. For the following results, we consider the Gilbert-Elliot channel, with $p_{11,1}=p_{11,2}=0.9$, and $p_{01,1}=p_{01,2}=0.6$. The length of the frame, $T$, is $20$ time slots, and the number of arrived packets at the beginning of every frame, $K$, is equal to $15$ packets. We consider that the maximum value of the \ac{AoI}, $A_\text{max} = 20$. The timely throughput requirements are $q=12$ packets/frame or $q/T=0.6$ packet/slot. We run each experiment for $0.5\times 10^6$ time slots, and we use MATLAB to perform our simulations.

In Fig. \ref{Fig:AoIThroughput}, we provide the average value of \ac{AoI} for different values of $V$ as well as the convergence of the timely throughput constraints. In Fig. \ref{fig:ageallvge}, we compare the average value of \ac{AoI} of five sample paths with that of one sample-path. We observe that the values are quite close to each other. Therefore, the algorithm offers high performance regarding robustness. Furthermore, it is shown that the \ac{AoI} reaches its minimum value even for small values of $V$. We see that for values of $V$ larger than $5$ the change of the value of the average \ac{AoI} is negligible. On the other hand, we observe that the convergence of the algorithm regarding the timely throughput constraints changes dramatically as $V$ increases, as shown in  Fig. \ref{fig:thrconvlongge}. 

\begin{figure}
\centering
\begin{subfigure}{0.4\textwidth}
\centering
\includegraphics[width=1\linewidth]{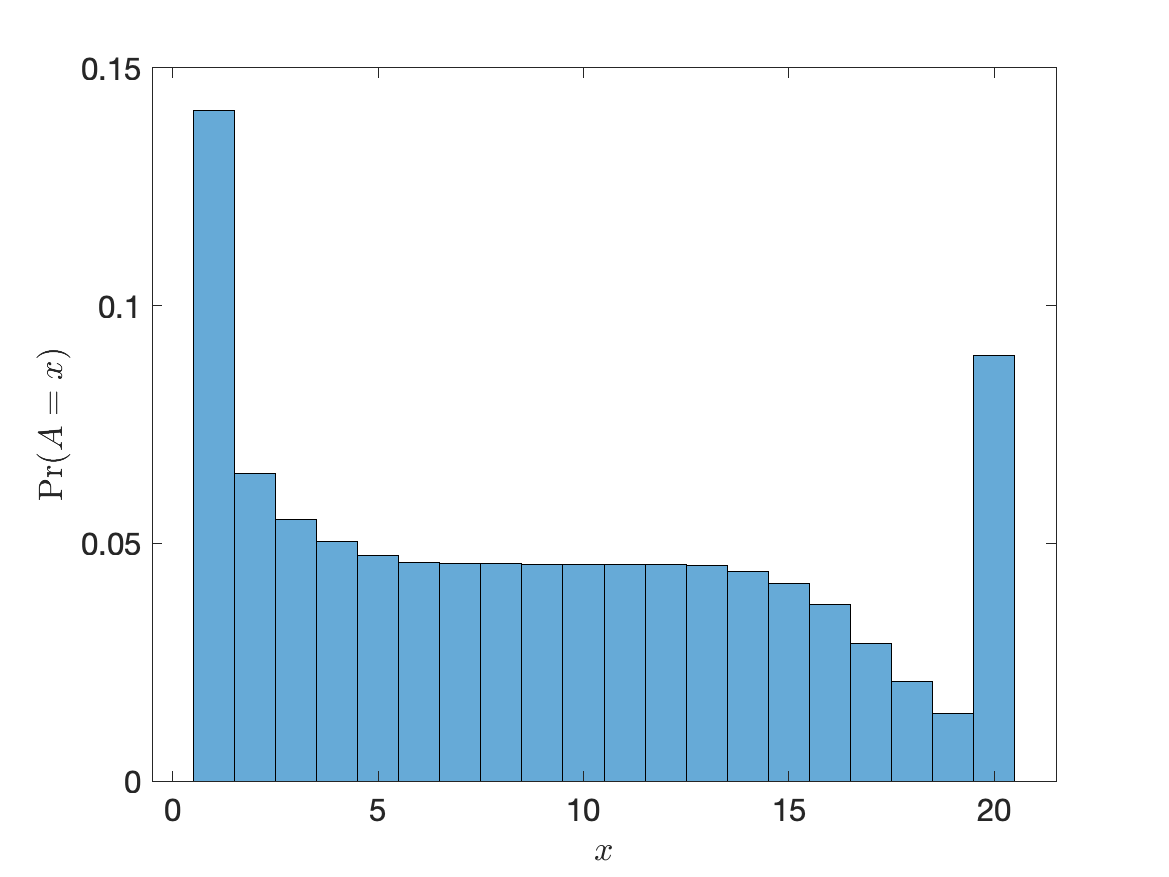}
\caption{$V=0$.}
\label{fig:agedistrv0ge}
\end{subfigure}
\begin{subfigure}{0.4\textwidth}
\centering
\includegraphics[width=1\linewidth]{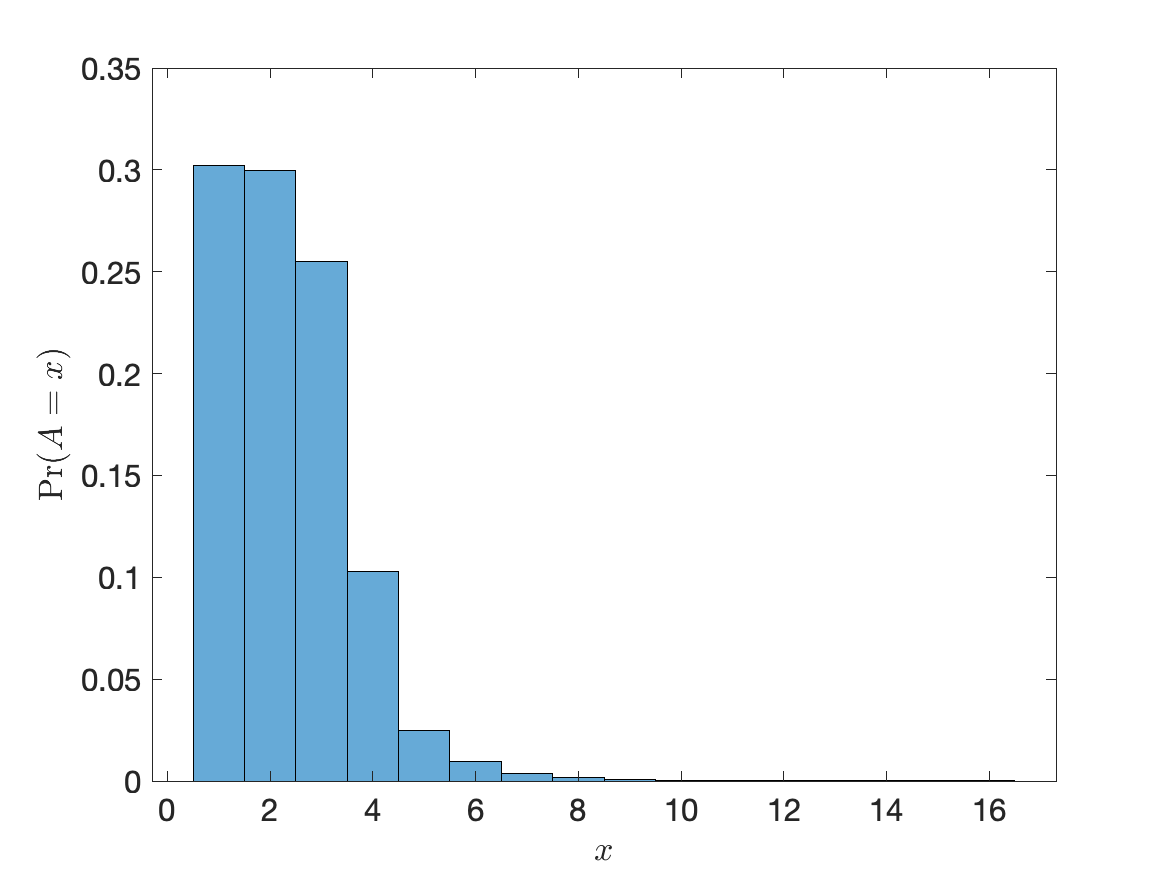}
\caption{$V=5$.}
\label{fig:agedistrv5ge}
\end{subfigure}
\begin{subfigure}{0.4\textwidth}
\centering
\includegraphics[width=1\linewidth]{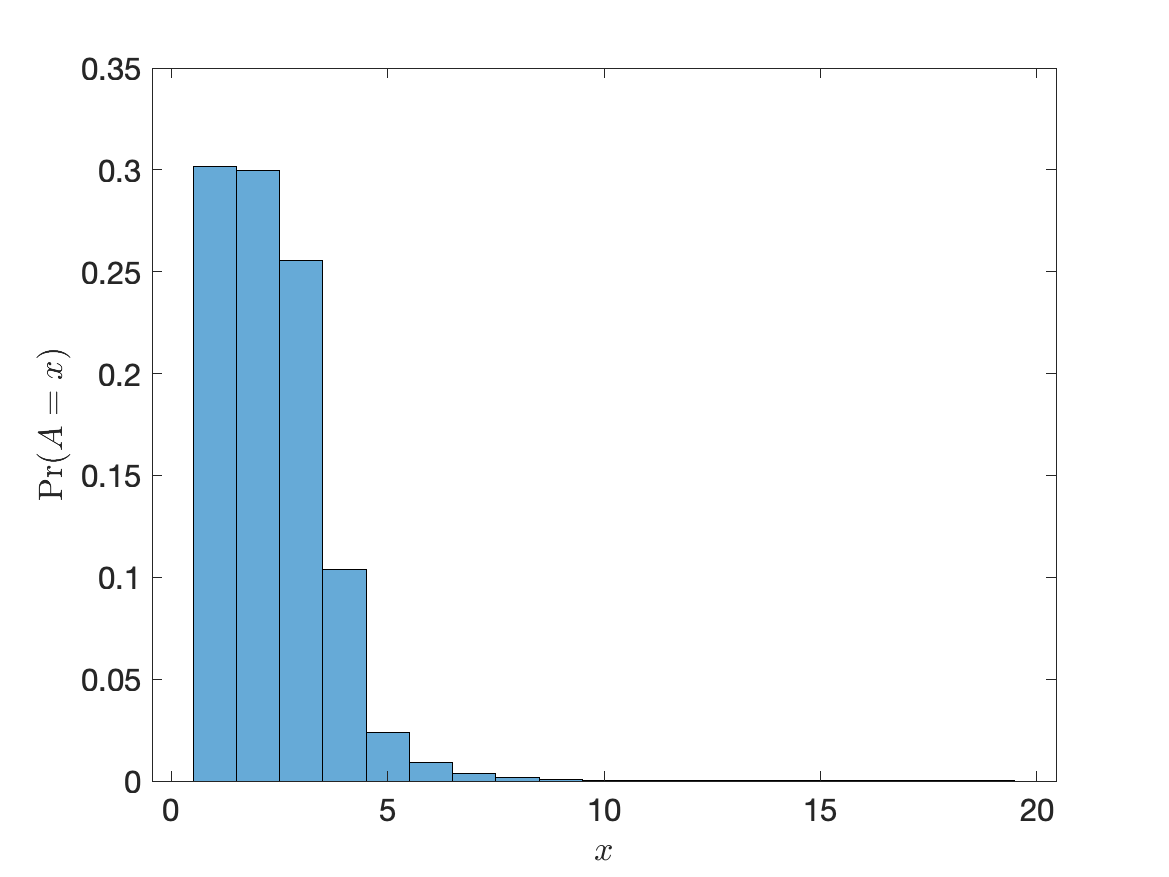}
\caption{$V=10$.}
\label{fig:agedistrv10ge}
\end{subfigure}
\begin{subfigure}{0.4\textwidth}
\centering
\includegraphics[width=1\linewidth]{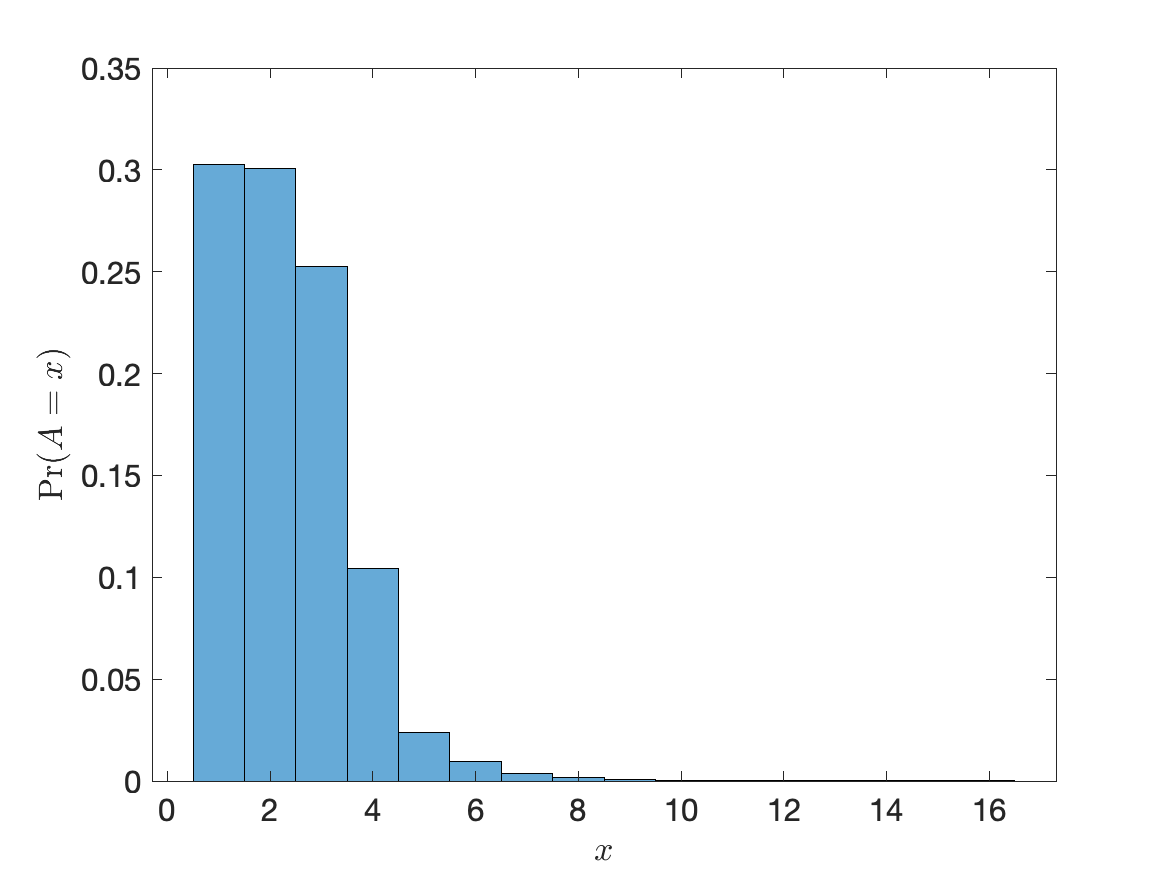}
\caption{$V=100$.}
\label{fig:agedistrv100ge}
\end{subfigure}
\caption{The distribution of  \ac{AoI} for different values of $V$.}
\label{fig:distribtionAoI}
\end{figure}
\begin{figure}[t!]
\begin{subfigure}{0.49\textwidth}
\centering
\includegraphics[width=0.8\linewidth]{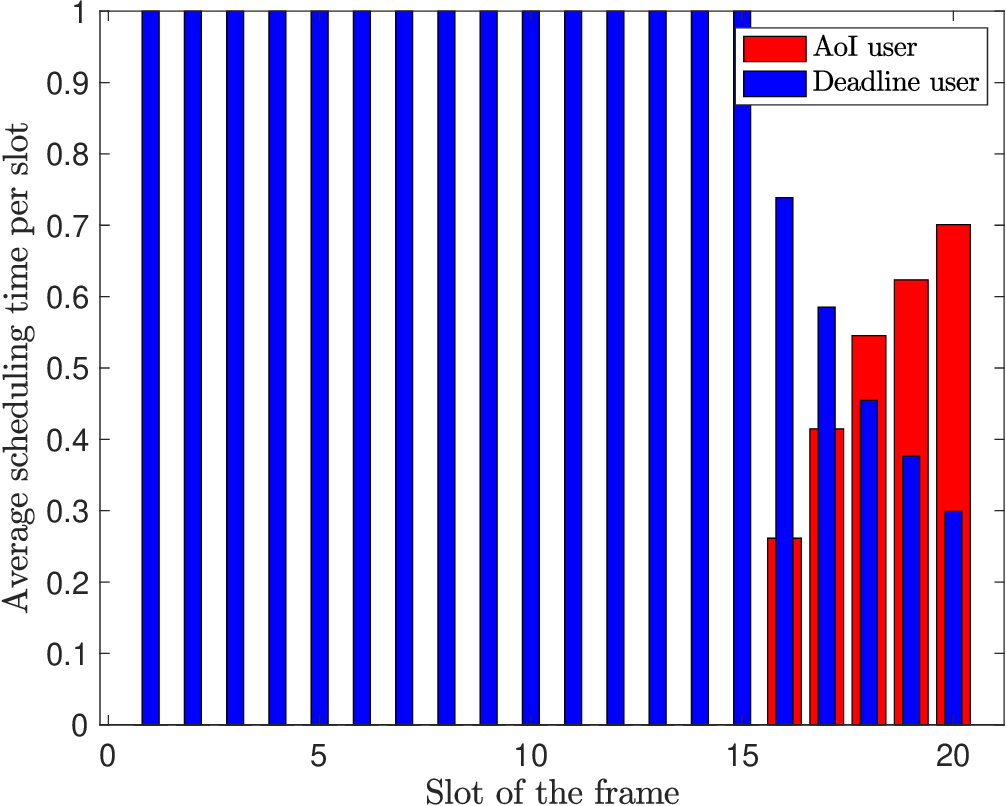}
\caption{$V=0$.}
\label{fig:perageschedv0ge}
\end{subfigure}
\begin{subfigure}{0.49\textwidth}
\centering
\includegraphics[width=0.8\linewidth]{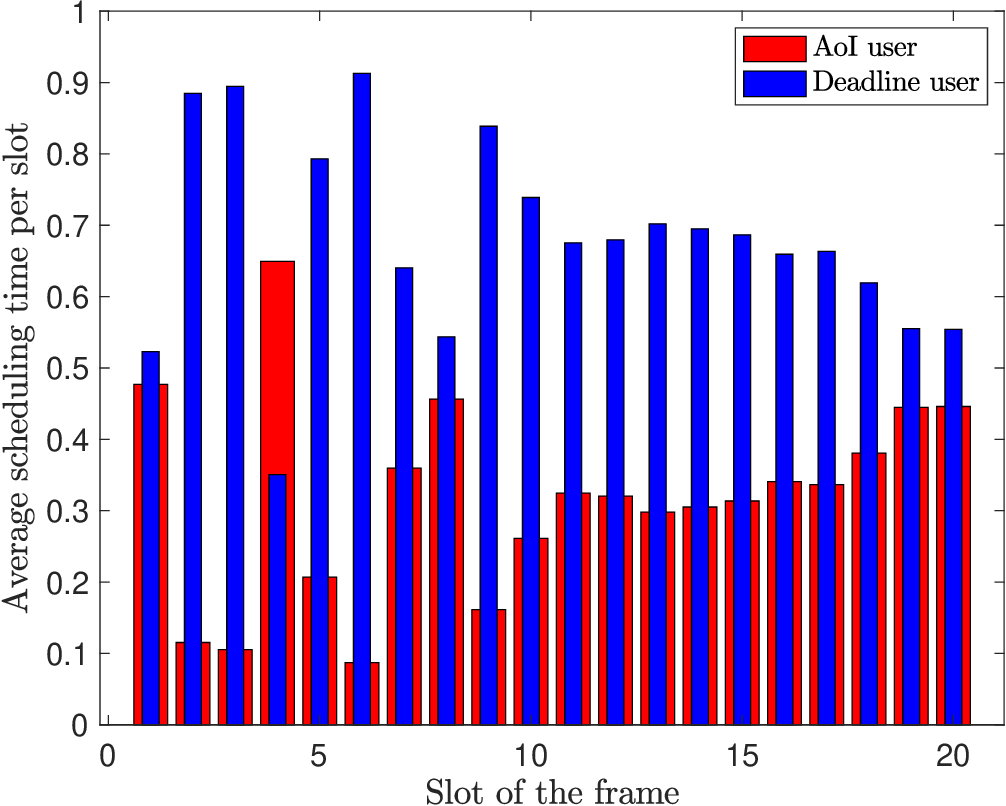}
\caption{$V=5$.}
\label{fig:perageschedv2ge}
\end{subfigure}
\begin{subfigure}{0.49\textwidth}
\centering
\includegraphics[width=0.8\linewidth]{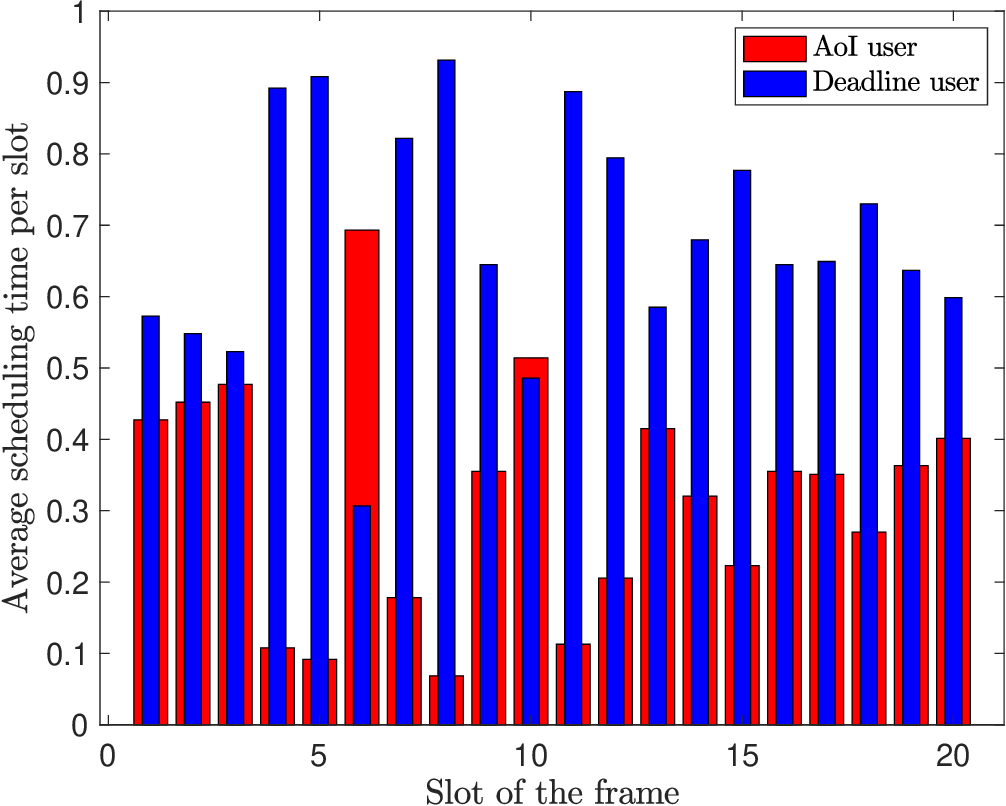}
\caption{$V=10$.}
\label{fig:perageschedv10ge}
\end{subfigure}
\begin{subfigure}{0.49\textwidth}
\centering
\includegraphics[width=0.8\linewidth]{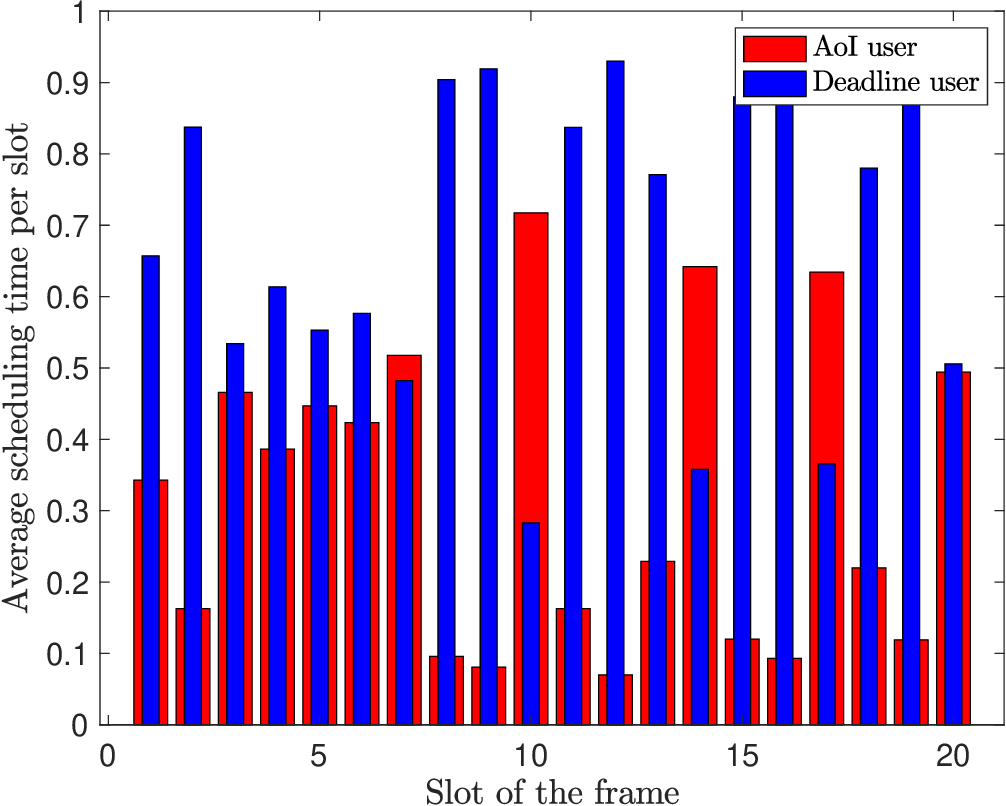}
\caption{$V=100$.}
\label{fig:perageschedv100ge}
\end{subfigure}
\caption{Scheduling time percentage per slot within a frame.}
\label{fig:schedulingtimes}
\end{figure}

In Fig. \ref{fig:distribtionAoI}, we provide results for the distribution of the \ac{AoI} obtained by simulations for different values of $V$. In Fig. \ref{fig:agedistrv0ge}, we observe that the \ac{AoI} reaches the maximum value because its weight is zero. Therefore, the algorithm schedules user $2$ as long as its queue is non empty. After emptying the queue of user $2$, the scheduler schedules user $1$ if there are remaining slots. For values larger or equal than $5$, in Figs. \ref{fig:agedistrv5ge}, \ref{fig:agedistrv10ge}, \ref{fig:agedistrv100ge}, we observe that the \ac{AoI} never reaches its maximum value. Instead, the values of \ac{AoI} fluctuate mainly in the range of $1-5$.

In Fig. \ref{fig:schedulingtimes}, we provide results that show the scheduling time percentage per slot within a frame for each user. In Fig. \ref{fig:perageschedv0ge}, the value of $V$ is $0$. 
For the first $15$ slots the scheduler schedules only user $2$. The number of packets is $15$, therefore, the scheduler needs at least $15$ slots to empty the queue of user $2$ because a packet may be failed to be transmitted and the transmitter has to retransmit it. After emptying the queue of user $2$, if there are remaining slots, the scheduler allocates slots to user $1$. In Figs. \ref{fig:perageschedv2ge}, \ref{fig:perageschedv10ge}, \ref{fig:perageschedv100ge}, we observe that for different values of $V$ the scheduling time for every user changes. However, as we observe in Fig. \ref{fig:ageallvge}, the values of average \ac{AoI} are quite close to each other. We observe that for $V=5$, the scheduling time for user $1$ is spread within the frame. That means that the percentage of scheduling time does not change significantly from slot to slot. On the other hand, for larger values of $V$, we observe that the percentage of the scheduling time for user $1$ changes from slot to slot, especially after the $10^\text{th}$ slot and for $V=100$  because the \ac{AoI} is multiplied by a large weight and if the value of \ac{AoI} starts increasing as time passes by the corresponding term becomes quite large. Therefore, the scheduler schedules the user $1$ in order to minimize the objective function. 
\subsection{Frame-based \ac{DPP} Vs slot-based \ac{DPP}}
\begin{figure}[h!]
\centering
\includegraphics[scale=0.5]{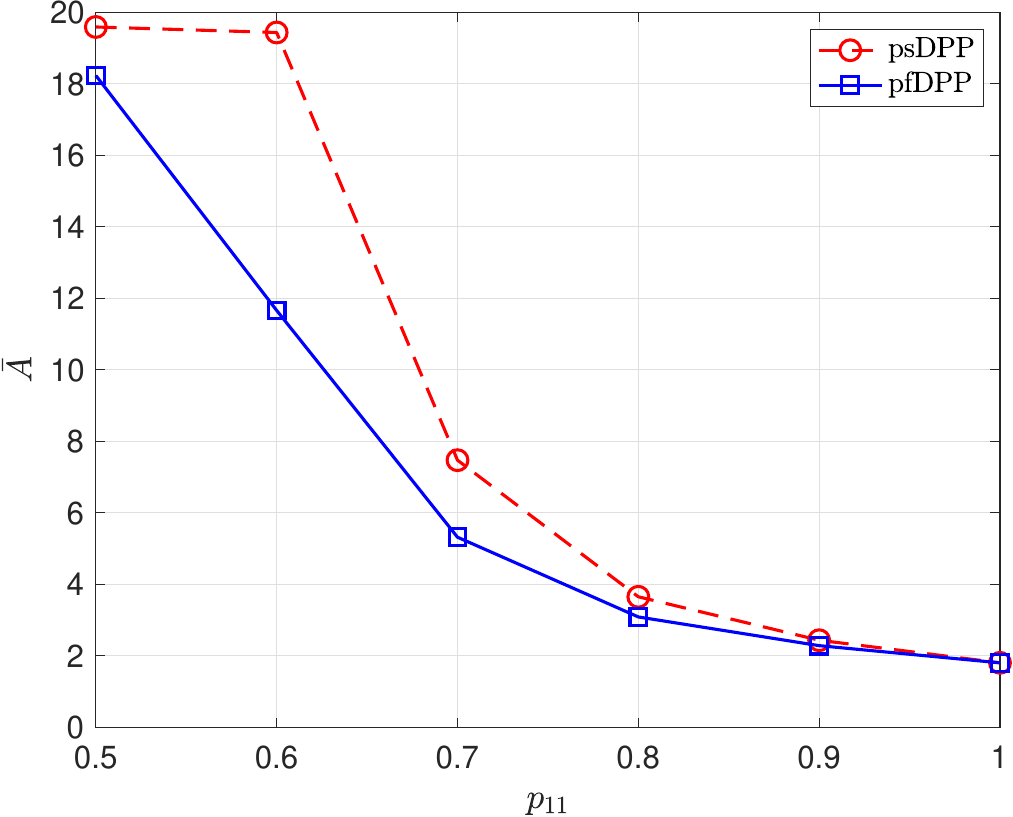}
\caption{pfDPP Vs psDPP.}
\label{fig:pfDPPvspsDPP}
\end{figure}
\revtwo{
In this subsection, we compare the performance of the near-optimal algorithm (\ac{pfDPP}) with the performance of (\ac{psDPP}). The transition probability for both users, $p_{\TC{01}}$, is equal to $0.6$. We obtain results for the average \ac{AoI} by changing the transition probability, $p_{11}$, with values between $0.5$ to $1$ with step $0.1$. We select $V=200$.}

\revtwo{
In Fig. \ref{fig:pfDPPvspsDPP}, we observe that for low values of transition probability $p_{11}$, the near-optimal algorithm (\ac{pfDPP}), has significantly better performance than the low-complexity algorithm (\ac{psDPP}). However, we see that for values of $p_{11}$ higher than $0.7$, the performance of the two algorithms is quite close. Therefore, from the results, we obtain that the statistics of the channel play an important role on how the decisions are made based on the look-ahead principle.}

\section{Conclusion \& Future Work}
In this work, we considered a wireless network consisting of time-critical users with different requirements under uncertain environments. We studied how an \ac{AoI}-oriented user and a deadline-constrained user can share the same resources to satisfy their requirements. To this end, we formulated a stochastic optimization problem for minimizing the average \ac{AoI} while satisfying the timely-throughput constraints which is a \ac{CMDP} problem. In order to solve the problem, we utilized tools from Lyapunov optimization and \ac{MDP}. With this approach, we reduced the \ac{CMDP} to an unconstrained
weighted stochastic shortest path problem. We implemented backward dynamic programming to solve the unconstrained problem. Simulation results showed that the timely-throughput constraint is satisfied while minimizing the average \ac{AoI}. Furthermore, we provided the trade-off between the minimum value of \ac{AoI} and the convergence of the average constraint. Also, based on the Lyapunov optimization theory, we provide a sub-optimal low-complexity algorithm that guarantees that the timely-throughput of user $2$ is above the required threshold. 
As a future work, an interesting direction is a case with variable lengths of the frames and random packet arrivals of the deadline-constrained user. Also, a realistic scenario is the multi-user scenario that will make our problem even more interesting but more difficult to be solved. In this case, backward dynamic programming has prohibitive complexity and approximation algorithms could give a good solution.

\begin{appendices}
\section{Drift Bound: Proof of Lemma 2}\label{Appendix:Proof of Lemma 2}
\begin{proof}
By utilizing $(\max \left[Q-b,0 \right]+A)^2 \leq  Q^2+A^2+b^2 + 2Q(A-b)\text{,}$ we get 
\begin{align}
	Z^2(t+1) \leq Z^2(t) + q^2 + d_2^2(t) + 2Z(t)(q-d_2(t)) \text{,}
\end{align}
by adding and substituting  the term $2Z(t_m)(q-d_2(t))$, we get 
\begin{align}\nonumber
	Z^2(t+1) 
	 &\leq Z^2(t) + q^2 +1 + 2Z(t_m)(q-d_2(t)) \\
	 &+ 2(t-t_m)\text{,}
\end{align}
diving by $2$, using telescoping sums, and using that $\sum\limits_{\tau=t_m}^{t_m+T-1} (\tau -t_m) = \frac{T(T-1)}{2}$, we get, 
\begin{align}
	&\frac{Z^2(t_m+T) - Z^2(t_m)}{2} \\
	&\leq \frac{Tq^2 + T(T-1)}{2} + Z(t_m) \sum\limits_{\tau=t_m}^{t_m+T-1} (q-d_2(\tau))\text{.}
\end{align} 
Finally, by taking conditional expectations, we get the result, 
\begin{align}
	\label{ineq:BouundLyapunov}
	\Delta(L(Z(t_m))) \leq B + \mathbb{E} \left[ G(t_m) | Z(t_m)\right]\text{.}
\end{align}
\end{proof}
\section{Virtual Queue Stability}
From \eqref{BoundVirtualQueue}, we get 
$
\label{BoundD}
\frac{1}{R} \sum\limits_{m=0}^R \mathbb{E} \left[Z(t_m)\right] \leq D\text{,}
$
where $D$ is a bounded constant, if it holds, then 
$
\lim\limits_{t\rightarrow \infty} \sup \frac{1}{t} \sum\limits_{\tau=0} ^{t-1} \mathbb{E}\left[Z(\tau)\right] < \infty \text{.}
$
\begin{proof}
Since every frame is at least one slot, we have
\begin{align}
	\sum\limits_{\tau=0}^{R-1} Z(\tau)\leq \sum\limits_{r=0} ^{R-1} \sum\limits_{\tau = t_m}^{t_m+T-1} Z(\tau)\leq \sum\limits_{r=0}^{R-1} T \left[Z(t_m) + T\right]\text{.}
\end{align}
By taking the expectations, we get 
$\sum\limits_{\tau=0}^{R-1} \mathbb{E} \left[Z(\tau)\right] \leq RT^2 + DT\text{.}$
By taking the limit $R\rightarrow \infty$, and diving by $R$, we get,
$\lim\limits_{R\rightarrow \infty} \sup \frac{1}{R} \sum\limits_{\tau =0}^{R-1} \mathbb{E} \left[ Z(\tau)\right] \leq T^2 +DT\text{.}$
That proves the result.
\end{proof}

\section{Proof of Theorem 2}\label{Appendix: Proof of Theorem 2}	
\begin{proof}
Part I - Queue bound. From \eqref{ineq:BouundLyapunov}, we have 
\begin{align}\label{ineq:BouundLyapunov2}
	&\Delta(t_m) + \mathbb{E}\left[V\sum\limits_{\tau=t_m}^{t_m+T-1} A(\tau)|Z(t)\right]  
	\leq B \\
	&+ \mathbb{E} \left[G_\text{opt}(t_m)  + V\sum\limits_{\tau=t_m}^{t_m+T-1} A_\text{opt}(\tau)  \right] + C + \delta Z(t_m) +V\delta\text{.}
\end{align}
Note that $A(\tau) \geq A_\text{opt}(\tau)$ therefore, we have
\begin{align}\nonumber
	\Delta(L(Z(t_m)))  & \leq B  + \mathbb{E} \left[ G_\text{opt}(t_m) \right] + \mathbb{E} \left[V\sum\limits_{\tau=t}^{t_m+T-1}  A_\text{opt}(\tau)  \right] \\\nonumber
	& + C + \delta Z(t_m) +V\delta \\ \nonumber
	\Delta(L(Z(t_m))) & \leq B +\mathbb{E}\left[G_\text{opt}(t_m) \right] + VT\bar{A}_{\text{opt}}\\\label{ineq:DriftStar}
	& + C + \delta Z(t_m)  + V\delta\text{.}
\end{align}
Now consider that the policy $\mathbf{u}_2(t)$ from Assumption 2,
\begin{align}
	\label{ineq:G}
	\mathbb{E} \left[G(t_m) |Z(t)\right] \leq -\epsilon T Z(t_m)\text{,}
\end{align}
Substituting \eqref{ineq:G} into \eqref{ineq:DriftStar}, we get
\begin{align}\nonumber
	\Delta(L(Z(t_m))) & \leq B - \epsilon T Z(t_m) \\
	&+ VT\bar{A}_{\text{opt}}+C + \delta Z(t_m) +V\delta \\ \nonumber
	\mathbb{E} \left[L(t_m+1)\right] - \mathbb{E}\left[L(t_m)\right] & \leq B + C + VT\bar{A}_{\text{opt}} \\
	& + (\delta - \epsilon T) Z(t_m) + V\delta \text{.}
\end{align}
Summing over $m\in \left\{0, \ldots, R-1\right\}$ and using the fact that $\mathbb{E} \left[L(t_0)\right]=0$, and utilizing the telescoping sums, we get  
\begin{align}\nonumber
	\frac{\mathbb{E}\left[L(t_m)\right]}{R} & \leq B + C + V(T\bar{A}_{\text{opt}} + \delta) \\ 
	&+ \frac{\delta -\epsilon T}{R} \sum\limits_{\tau=t_m}^{t_m+T-1} Z(\tau) \\ \nonumber
	\frac{\epsilon T - \delta}{R} \sum\limits_{\tau = t_m}^{t_m + T -1} Z(t_m) &\leq \frac{-\mathbb{E}\left[L(t_m)\right]}{R} + B +C \\&+V(T\bar{A}_{\text{opt}} + \delta)\text{,}
\end{align}
By ignoring the negative part, we get 
$
\frac{1}{R} \sum\limits_{r=0}^R \mathbb{E} \left[Z(t_m)\right] \leq \frac{B +C +V(T\bar{A}_{\text{opt}} + \delta)}{\epsilon T - \delta}\text{.}
$
In addition, by taking $R\rightarrow \infty$, we obtain that
\begin{align}\label{BoundVirtualQueue}
	\lim\limits_{R\rightarrow \infty} \sup \frac{1}{R} \sum\limits_{r=0}^R \mathbb{E} \left[Z(t_m)\right] \leq \frac{B +C +V(T\bar{A}_{\text{opt}} + \delta)}{\epsilon T - \delta}
\end{align}
From \eqref{BoundVirtualQueue}, we obtain that the virtual queue is bounded above, and therefore, the virtual queue is strongly stable. Therefore, the constraints in \eqref{opt:constr} are satisfied (See Appendix A). 

Part II - Approximation theorem. Define the probability $\gamma = \frac{\delta}{\epsilon T}$. This is a valid probability by assumption ($\epsilon T >\delta$). We consider a policy $\omega$ performed over the frame $\tau \in \left\{t_m, \ldots, t_m+T-1\right\}$. The policy $\omega$ is a randomized mixture from Assumptions $1$ and $2$. At the beginning of each frame, flip a coin with probabilities $\gamma$ and $1-\gamma$, and apply one of the policies as follows
1) With probability $\gamma$, apply policy $\omega_2$ from Assumption $2$, 2) with probability $1-\gamma$, apply policy $\omega_1$ from Assumption $1$.
From \eqref{eq:Ass1Obj}, we have 
\begin{align}
	\label{AStarBound}
	\mathbb{E} \left[\sum\limits_{\tau=t_m}^{t_m+T-1} A(\tau)\right] \leq T(\gamma A_{\text{max}} + (1-\gamma)A_{\text{opt}})\text{,}
\end{align}
we also have from \eqref{eq:Ass2const}, 
\begin{align}
	\label{BoundMinusDelta}
	\mathbb{E} \left[\sum\limits_{\tau=t_m}^{t_m+T-1} q - d(\tau)\right] \leq -\gamma \epsilon T = -\delta\text{,}
\end{align}
plugging \eqref{AStarBound}, \eqref{BoundMinusDelta}, into \eqref{ineq: CdApprox}, we get 
\begin{align}\nonumber
 &	\Delta(L(Z(t_m))) +\mathbb{E} \left[ V\sum\limits_{\tau=t_m}^{t_m+T-1} A(\tau) | Z(t_m)\right] \\
	&  \leq B + T(\gamma(A_{\text{max}} -1 )) + (1-\gamma)A^{\text{opt}} + C + V\delta \\\nonumber
	&\mathbb{E}\left[L(t_m+1)\right] - \mathbb{E}\left[L(t_m)\right] + V\mathbb{E} \left[\sum\limits_{\tau=t_m}^{t_m+T-1}A(\tau)\right] \\	&\leq B +VT(\gamma(A_{\text{max}} - 1)) + C + V\delta\text{.}
\end{align}
Summing over $m \in \left\{0, \ldots, M-1\right\}$ and dividing by $VMT$ gives the following, $\forall M>0$,
\begin{align}\nonumber
	\frac{1}{MT} \mathbb{E} \left[\sum\limits_{\tau=0}^{t_M-1} A(\tau)\right] & \leq \frac{B}{VT} +  (\gamma(A_{\text{max}}-1)\\
	& + (1-\gamma)A^{\text{opt}} ) + \frac{C}{VT}  + \frac{\delta}{T}\text{.} 
\end{align}
Finally, by taking $M\rightarrow \infty$, we get the result.
\end{proof}

\end{appendices}

\bibliography{MyBibAgeTimelyTh}

\begin{thebibliography}{10}

\bibitem{abd2019role}
M.~A. Abd-Elmagid, N.~Pappas, and H.~S. Dhillon, ``On the role of age of
  information in the internet of things,'' {\em IEEE Comm. Mag.}, vol.~57,
  no.~12, pp.~72--77, 2019.

\bibitem{shreedhar2019age}
T.~Shreedhar, S.~K. Kaul, and R.~D. Yates, ``An age control transport protocol
  for delivering fresh updates in the {I}nternet-of-{T}hings,'' {\em in Proc.
  IEEE WoWMoM}, pp.~1--7, 2019.

\bibitem{TheoryQoS2009}
I.~. {Hou}, V.~{Borkar}, and P.~R. {Kumar}, ``A theory of {Q}o{S} for
  wireless,'' {\em in Proc. IEEE INFOCOM 2009}, pp.~486--494, 2009.

\bibitem{kosta2017age}
A.~Kosta, N.~Pappas, and V.~Angelakis, ``Age of information: A new concept,
  metric, and tool,'' {\em Foundations and Trends in Networking}, vol.~12,
  no.~3, pp.~162--259, 2017.

\bibitem{sun2019age}
Y.~Sun, I.~Kadota, R.~Talak, and E.~Modiano, ``Age of information: A new metric
  for information freshness,'' {\em Synthesis Lectures on Communication
  Networks}, vol.~12, no.~2, pp.~1--224, 2019.

\bibitem{kaulyates2012real}
S.~Kaul, R.~Yates, and M.~Gruteser, ``Real-time status: How often should one
  update?,'' {\em in Proc. IEEE INFOCOM}, 2012.

\bibitem{shakkottai2002scheduling}
S.~Shakkottai and R.~Srikant, ``Scheduling real-time traffic with deadlines
  over a wireless channel,'' {\em Wireless Networks}, vol.~8, no.~1,
  pp.~13--26, 2002.

\bibitem{cui2012survey}
Y.~Cui, V.~K. Lau, R.~Wang, H.~Huang, and S.~Zhang, ``A survey on delay-aware
  resource control for wireless systems—{L}arge deviation theory, stochastic
  {L}yapunov drift, and distributed stochastic learning,'' {\em IEEE Trans.
  Info Theory}, vol.~58, no.~3, pp.~1677--1701, 2012.

\bibitem{you2018resource}
L.~You, Q.~Liao, N.~Pappas, and D.~Yuan, ``Resource optimization with flexible
  numerology and frame structure for heterogeneous services,'' {\em IEEE Com.
  Letters}, vol.~22, no.~12, pp.~2579--2582, 2018.

\bibitem{ManosWiOpt2017}
E.~{Fountoulakis}, N.~{Pappas}, Q.~{Liao}, V.~{Suryaprakash}, and D.~{Yuan},
  ``An examination of the benefits of scalable {TTI} for heterogeneous traffic
  management in {5G} networks,'' {\em in Proc. WiOpt}, pp.~1--6, 2017.

\bibitem{ElAzzouni2020}
S.~{ElAzzouni}, E.~{Ekici}, and N.~{Shroff}, ``Is deadline oblivious scheduling
  efficient for controlling real-time traffic in cellular downlink systems?,''
  {\em in Proc. IEEE INFOCOM}, pp.~49--58, 2020.

\bibitem{tsanikidis2021power}
C.~Tsanikidis and J.~Ghaderi, ``On the power of randomization for scheduling
  real-time traffic in wireless networks,'' {\em IEEE/ACM Trans. Net.}, 2021.

\bibitem{destounis2018scheduling}
A.~Destounis, G.~S. Paschos, J.~Arnau, and M.~Kountouris, ``Scheduling {URLLC}
  users with reliable latency guarantees,'' {\em Proc. WiOpt}, pp.~1--8, 2018.

\bibitem{tsanikidis2022randomized}
C.~Tsanikidis and J.~Ghaderi, ``Randomized scheduling of real-time traffic in
  wireless networks over fading channels,'' {\em IEEE/ACM Trans. Net.}, 2022.

\bibitem{lashgari2013timely}
S.~{Lashgari} and A.~S. {Avestimehr}, ``Timely throughput of heterogeneous
  wireless networks: Fundamental limits and algorithms,'' {\em IEEE Trans. Inf.
  Theory}, vol.~59, no.~12, pp.~8414--8433, 2013.

\bibitem{ManosGC2018}
E.~{Fountoulakis}, N.~{Pappas}, Q.~{Liao}, A.~{Ephremides}, and V.~{Angelakis},
  ``Dynamic power control for packets with deadlines,'' {\em Proc. IEEE
  GLOBECOM}, 2018.

\bibitem{TepedeTVT2018}
A.~E. {Ewaisha} and C.~{Tepedelenlioglu}, ``Optimal power control and
  scheduling for real-time and non-real-time data,'' {\em IEEE Trans. Vehic.
  Tech.}, vol.~67, no.~3, pp.~2727--2740, 2018.

\bibitem{fountoulakisITU}
E.~Fountoulakis, N.~Pappas, and A.~Ephremides, ``Dynamic power control for
  time-critical networking with heterogeneous traffic,'' {\em ITU J-FET}, March
  2021.

\bibitem{master2016power}
N.~Master and N.~Bambos, ``Power control for packet streaming with head-of-line
  deadlines,'' {\em Performance Evaluation}, vol.~106, pp.~1--18, 2016.

\bibitem{neely2013dynamic}
M.~J. Neely and S.~Supittayapornpong, ``Dynamic markov decision policies for
  delay constrained wireless scheduling,'' {\em IEEE Trans. on Automatic
  Control}, vol.~58, no.~8, pp.~1948--1961, 2013.

\bibitem{ModianoTWC2006}
A.~{Fu}, E.~{Modiano}, and J.~N. {Tsitsiklis}, ``Optimal transmission
  scheduling over a fading channel with energy and deadline constraints,'' {\em
  IEEE Trans. Wireless Commun.}, vol.~5, no.~3, pp.~630--641, 2006.

\bibitem{TNSM2020}
S.~{Saraswat}, H.~P. {Gupta}, T.~{Dutta}, and S.~K. {Das}, ``Energy efficient
  data forwarding scheme in fog-based ubiquitous system with deadline
  constraints,'' {\em IEEE Trans. Net. Service Man.}, vol.~17, no.~1,
  pp.~213--226, 2020.

\bibitem{mao2014optimal}
Z.~Mao, C.~E. Koksal, and N.~B. Shroff, ``Optimal online scheduling with
  arbitrary hard deadlines in multihop communication networks,'' {\em IEEE/ACM
  Trans. Net.}, vol.~24, no.~1, pp.~177--189, 2014.

\bibitem{tsanikidis2022online}
C.~Tsanikidis and J.~Ghaderi, ``Online scheduling and routing with end-to-end
  deadline constraints in multihop wireless networks,'' {\em Proc. ACM
  MobiHoc}, pp.~11--20, 2022.

\bibitem{gu2021asymptotically}
Y.~Gu, B.~Liu, and X.~Shen, ``Asymptotically optimal online scheduling with
  arbitrary hard deadlines in multi-hop communication networks,'' {\em IEEE/ACM
  Trans. Net.}, vol.~29, no.~4, pp.~1452--1466, 2021.

\bibitem{deng2019online}
H.~Deng, T.~Zhao, and I.-H. Hou, ``Online routing and scheduling with capacity
  redundancy for timely delivery guarantees in multihop networks,'' {\em
  IEEE/ACM Trans. Net.}, vol.~27, no.~3, pp.~1258--1271, 2019.

\bibitem{talak2020age}
R.~Talak and E.~Modiano, ``Age-delay tradeoffs in queueing systems,'' {\em IEEE
  Trans. Info Theory}, 2020.

\bibitem{kadota2021minimizing}
I.~Kadota and E.~Modiano, ``Minimizing the age of information in wireless
  networks with stochastic arrivals,'' {\em IEEE Trans. Mob. Comp.}, vol.~20,
  no.~3, pp.~1173--1185, 2021.

\bibitem{kadota2019scheduling}
I.~Kadota, A.~Sinha, and E.~Modiano, ``Scheduling algorithms for optimizing age
  of information in wireless networks with throughput constraints,'' {\em
  IEEE/ACM Trans. Net.}, vol.~27, no.~4, pp.~1359--1372, 2019.

\bibitem{zhou2019joint}
B.~Zhou and W.~Saad, ``Joint status sampling and updating for minimizing age of
  information in the internet of things,'' {\em IEEE Trans. Commun.}, vol.~67,
  no.~11, pp.~7468--7482, 2019.

\bibitem{abd2020aoi}
M.~A. Abd-Elmagid, H.~S. Dhillon, and N.~Pappas, ``Ao{I}-optimal joint sampling
  and updating for wireless powered communication systems,'' {\em IEEE Trans
  Vehic. Tech.}, vol.~69, no.~11, pp.~14110--14115, 2020.

\bibitem{BedewyOptimalSampling}
A.~M. {Bedewy}, Y.~{Sun}, S.~{Kompella}, and N.~B. {Shroff}, ``Optimal sampling
  and scheduling for timely status updates in multi-source networks,'' {\em
  IEEE Trans. Theory}, pp.~1--1, 2021.

\bibitem{fountoulakis2021joint}
E.~Fountoulakis, N.~Pappas, M.~Codreanu, and A.~Ephremides, ``Optimal sampling
  cost in wireless networks with age of information constraints,'' {\em in
  Proc. IEEE INFOCOM}, pp.~918--923, 2020.

\bibitem{yao2020age}
G.~Yao, A.~M. Bedewy, and N.~B. Shroff, ``Age minimization transmission
  scheduling over time-correlated fading channel under an average energy
  constraint,'' {\em arXiv preprint arXiv:2012.02958}, 2020.

\bibitem{ceran2019reinforcement}
E.~T. Ceran, D.~G{\"u}nd{\"u}z, and A.~Gy{\"o}rgy, ``Reinforcement learning to
  minimize age of information with an energy harvesting sensor with {HARQ} and
  sensing cost,'' {\em in Proc. IEEE INFOCOM}, pp.~656--661, Apr. 2019.

\bibitem{GstamIoTJ}
G.~{Stamatakis}, N.~{Pappas}, and A.~{Traganitis}, ``Optimal policies for
  status update generation in an {I}o{T} device with heterogeneous traffic,''
  {\em IEEE Internet Things J.}, 2020.

\bibitem{ZhengOpenJournal}
Z.~Chen, N.~Pappas, E.~Björnson, and E.~G. Larsson, ``Optimizing information
  freshness in a multiple access channel with heterogeneous devices,'' {\em
  IEEE OJ-COMS}, vol.~2, pp.~456--470, 2021.

\bibitem{SPAWC2019}
N.~{Pappas} and M.~{Kountouris}, ``Delay violation probability and age of
  information interplay in the two-user multiple access channel,'' {\em Proc.
  IEEE SPAWC}, 2019.

\bibitem{fountoulakis2022information}
E.~Fountoulakis, T.~Charalambous, N.~Nomikos, A.~Ephremides, and N.~Pappas,
  ``Information freshness and packet drop rate interplay in a two-user
  multi-access channel,'' {\em Journal of Communications and Networks}, 2022.

\bibitem{sun2021age}
J.~Sun, L.~Wang, Z.~Jiang, S.~Zhou, and Z.~Niu, ``Age-optimal scheduling for
  heterogeneous traffic with timely throughput constraints,'' {\em IEEE
  J-SAC)}, vol.~39, no.~5, pp.~1485--1498, 2021.

\bibitem{moltafet2019power}
M.~Moltafet, M.~Leinonen, M.~Codreanu, and N.~Pappas, ``Power minimization for
  age of information constrained dynamic control in wireless sensor networks,''
  {\em IEEE Trans. Comm.}, vol.~70, no.~1, pp.~419--432, 2021.

\bibitem{altman1999constrained}
E.~Altman, {\em Constrained Markov decision processes}, vol.~7.
\newblock CRC Press, 1999.

\bibitem{salodkar2008line}
N.~Salodkar, A.~Bhorkar, A.~Karandikar, and V.~S. Borkar, ``An on-line learning
  algorithm for energy efficient delay constrained scheduling over a fading
  channel,'' {\em IEEE JSAC}, vol.~26, no.~4, pp.~732--742, 2008.

\bibitem{djonin2007q}
D.~V. Djonin and V.~Krishnamurthy, ``Q-learning algorithms for constrained
  markov decision processes with randomized monotone policies: Application to
  mimo transmission control,'' {\em IEEE Trans. Signal Proc.}, vol.~55, no.~5,
  pp.~2170--2181, 2007.

\bibitem{NeelyBook}
M.~J. {Neely}, {\em Stochastic Network Optimization with Application to
  Communication and Queueing Systems}.
\newblock Morgan \& Claypool, 2010.

\bibitem{bertsekas1997nonlinear}
D.~P. Bertsekas, ``Nonlinear programming,'' {\em Journal of the Operational
  Research Society}, vol.~48, no.~3, pp.~334--334, 1997.

\bibitem{meyn2012markov}
S.~P. Meyn and R.~L. Tweedie, {\em Markov chains and stochastic stability}.
\newblock Springer Science \& Business Media, 2012.

\bibitem{bertsekas2011dynamic}
D.~P. Bertsekas {\em et~al.}, ``Dynamic programming and optimal control 3rd
  edition, volume ii,'' {\em Belmont, MA: Athena Scientific}, 2011.

\bibitem{powell2007approximate}
W.~B. Powell, {\em Approximate Dynamic Programming: Solving the curses of
  dimensionality}, vol.~703.
\newblock John Wiley \& Sons, 2007.

\end{thebibliography}
\bibliographystyle{ieeetr}

\end{document}